\documentclass[a4paper,fleqn]{cas-dc}
\usepackage[numbers]{natbib}
\usepackage{hyperref}
\usepackage{amsmath,amssymb,amsfonts}
\usepackage{amsthm}
\usepackage{mathtools}
\DeclarePairedDelimiter{\ceil}{\lceil}{\rceil}

\usepackage{algorithm, algpseudocode}
\usepackage{algorithmicx}

\makeatletter
\algnewcommand{\LineComment}[1]{\Statex \quad \quad \quad \quad \quad \quad \hskip\ALG@thistlm \(\triangleright\) #1}
\algnewcommand{\LineCommentNoIdent}[1]{\Statex \quad \quad \quad  \hskip\ALG@thistlm \(\triangleright\) #1}
\makeatother

\usepackage{etoolbox}
\usepackage{tikz}
\usetikzlibrary{tikzmark}
\usetikzlibrary{calc}

\errorcontextlines\maxdimen

\newcommand{\ALGtikzmarkcolor}{black}
\newcommand{\ALGtikzmarkextraindent}{4pt}
\newcommand{\ALGtikzmarkverticaloffsetstart}{-.5ex}
\newcommand{\ALGtikzmarkverticaloffsetend}{-.5ex}
\makeatletter
\newcounter{ALG@tikzmark@tempcnta}

\newcommand\ALG@tikzmark@start{%
    \global\let\ALG@tikzmark@last\ALG@tikzmark@starttext%
    \expandafter\edef\csname ALG@tikzmark@\theALG@nested\endcsname{\theALG@tikzmark@tempcnta}%
    \tikzmark{ALG@tikzmark@start@\csname ALG@tikzmark@\theALG@nested\endcsname}%
    \addtocounter{ALG@tikzmark@tempcnta}{1}%
}

\def\ALG@tikzmark@starttext{start}
\newcommand\ALG@tikzmark@end{%
    \ifx\ALG@tikzmark@last\ALG@tikzmark@starttext
    \else
        \tikzmark{ALG@tikzmark@end@\csname ALG@tikzmark@\theALG@nested\endcsname}%
        \tikz[overlay,remember picture] \draw[\ALGtikzmarkcolor] let \p{S}=($(pic cs:ALG@tikzmark@start@\csname ALG@tikzmark@\theALG@nested\endcsname)+(\ALGtikzmarkextraindent,\ALGtikzmarkverticaloffsetstart)$), \p{E}=($(pic cs:ALG@tikzmark@end@\csname ALG@tikzmark@\theALG@nested\endcsname)+(\ALGtikzmarkextraindent,\ALGtikzmarkverticaloffsetend)$) in (\x{S},\y{S})--(\x{S},\y{E});%
    \fi
    \gdef\ALG@tikzmark@last{end}%
}

\definecolor{R}{RGB}{0,150,0}
\newcommand{\secondbest}[1]{\textcolor{R}{#1}} 

\apptocmd{\ALG@beginblock}{\ALG@tikzmark@start}{}{\errmessage{failed to patch}}
\pretocmd{\ALG@endblock}{\ALG@tikzmark@end}{}{\errmessage{failed to patch}}
\makeatother

\usepackage{graphicx, subcaption}
\usepackage{textcomp}
\usepackage{mathtools}
\usepackage{multirow}
\usepackage{adjustbox}
\usepackage{xcolor}
\usepackage{pgfplots}
\usepackage{bm}      

\usepackage{float}

\usepackage{placeins}
\usepackage[none]{hyphenat} 
\usepackage{booktabs}
\usepackage{array}

\newcommand{\ie}{\textit{i.e.}}
\newcommand{\eg}{\textit{e.g.}}

\newtheorem{theorem}{Theorem}

\newtheorem{proposition}[theorem]{Proposition}

\begin{document}
	
\let\WriteBookmarks\relax
\def\floatpagepagefraction{1}
\def\textpagefraction{.001}
\shorttitle{The Megopolis Resampler}
\shortauthors{J. Chesser et~al.}

\title [mode = title]{The Megopolis Resampler: Memory Coalesced Resampling on GPUs}
\author[1]{Joshua A. Chesser}[orcid=0000-0002-8936-5537]
\ead{joshua.chesser@adelaide.edu.au}
\address[1]{School of Computer Science, University of Adelaide, Adelaide, SA 5005, Australia }
\author[1]{Hoa Van Nguyen}[orcid=0000-0002-6878-5102]
\ead{hoavan.nguyen@adelaide.edu.au}
\author[1]{Damith C. Ranasinghe}[orcid=0000-0002-2008-9255]
\cormark[1]
\ead{damith.ranasinghe@adelaide.edu.au}
\cortext[cor1]{Corresponding author}

\begin{abstract}
The resampling process employed in widely used methods such as Importance Sampling (IS), with its adaptive extension (AIS), are used to solve challenging problems requiring approximate inference; for example, non-linear, non-Gaussian state estimation problems. However, the re-sampling process can be computationally prohibitive for practical problems with real-time requirements. We consider the problem of developing highly parallelisable resampling algorithms for massively parallel hardware architectures of modern graphics processing units (GPUs) to accomplish real-time performance. We develop a new variant of the Metropolis algorithm---\textit{Megopolis}---that improves performance without requiring a tuning parameter or reducing resampling quality. The \textit{Megopolis} algorithm is built upon exploiting the memory access patterns of modern GPU units to reduce the number of memory transactions without the need for tuning parameters. Extensive numerical experiments on GPU hardware demonstrate that the proposed Megopolis algorithm is numerically stable and outperforms the original Metropolis algorithm and its variants---Metropolis-C1 and Metropolis-C2--in speed and quality metrics. Further, given the absence of open tools in this domain and facilitating fair comparisons in the future and supporting the signal processing community, we also \textit{open source} the complete project, including a repository of source code with \textit{Megopolis} and all other comparison methods.
\end{abstract}

\begin{keywords}
	Resampling  \sep Megopolis \sep Metropolis \sep Particle Filters \sep Sequential Monte Carlo \sep Importance Sampling \sep Adaptive Importance Sampling \sep GPU 
\end{keywords}
\maketitle

\section{Introduction}
A wide range of domains including multi-object tracking~\cite{arulampalam2002a,doucet2001introduction,hoa2019trackerbots}, physics \cite{laubenfels2005feynman}, financial economics~\cite{lopes2011particle},  and statistics~\cite{cappe2006inference} involves problems needing to model nonlinearity and non-Gaussianity to accurately estimate the state of a dynamic system, unknown parameters or functions from noisy measurements. Monte Carlo (MC) methods are statistical sampling-based techniques to solve such challenging or high-dimensional problems numerically~\cite{liu2008monte}. In particular, MC methods use a collection of random samples (so-called particles with associated weights) to approximate probability density functions. One of the common MC methods for sampling from a complex distribution is the Markov chain Monte Carlo (MCMC) method, which relies on constructing a Markov chain based on a stationary (equilibrium) distribution to generate desired samples~\cite{chib1995understanding,hastings1970monte}. 

A crucial alternative to the MCMC method is Importance Sampling (IS)~\cite{glynn1989importance} with its adaptive extension (AIS)~\cite{bugallo2017adaptive,cornuet2012adaptive}. These importance sampling methods are attractive because of their strong theoretical basis, wide applicability and ease of understanding. In essence, importance sampling approximates a probability distribution by: i) drawing samples from a proposal distribution (or the posterior distribution in a Bayesian context);  ii) computing sampled weights based on the difference between proposal distribution and the target distribution. For online and recursive estimation problems, practitioners can employ a sequential Monte Carlo (SMC) method~\cite{del2006sequential,doucet2000sequential}, a recursive generalisation of IS methods. Notably, IS, AIS and SMC methods are more popular because of their parallel implementation capabilities, especially on modern computing hardware. Importantly, at the heart of many of these algorithms is a \textit{resampling} procedure to minimise the particle degeneracy problem ~\cite{robert2004monte}.

Resampling is a process of replicating high weight particles and removing low weight particles while retaining an approximation of the probability density function. The resampling process plays a vital role in IS and SMC methods to prevent the problem of particle degeneracy, a phenomenon where all but few particles have negligible weights~\cite{arulampalam2002a,doucet2000sequential,robert2004monte}. However, resampling is a collective operation over all particle weights and is therefore computationally expensive. Importantly, increasing the number of particles improves estimation accuracy and is an unavoidable necessity for approximating high dimensional distributions in practical problems. Consequently, resampling is a bottleneck impeding the parallelisation of sampling algorithms such as IS, AIS and SMC for practical problems demanding the manipulation of large numbers of particles~\cite{dulger_memory_2018,hendeby_graphics_2007,hendeby_particle_2010}. Therefore, we consider the problem of developing highly parallelisable resampling algorithms for massively parallel hardware architectures of modern graphics processing units (GPUs) to address performance limitations. 

Common resampling algorithms include multinomial, residual, stratified and systematic \cite{bolic2005resampling,douc2005comparison,hol2006resampling,li2015resampling,sileshi_particle_2013}. These methods use the normalised particle weights to compute a particle's offspring, where the number of offspring is the number of times to duplicate that particle. Computing the normalised particle weights involves computing a cumulative sum over the particle weights known as a \textit{prefix sum}. Previous studies have focused on the parallelisation of the prefix sum and the particle selection stage of the algorithms  \cite{gong_parallel_2012,harris_optimizing_nodate,hendeby_graphics_2007,hendeby_particle_2010} to improve performance.  Recently, one notable approach is to utilise the \textit{monotonously} increasing nature of the prefix sum to \textit{parallelise} the systematic and stratified resampling algorithms by reindexing particles using multiple threads~\cite{nicely2019improved}. However, the cumulative sum can lead to numerical instabilities when single-precision representations are used for particle weights, and the number of particles is large. Although using double-precision can alleviate the problem, on contemporary hardware, single-precision performance is significantly faster than double-precision, making single-precision a more desirable choice for production scale algorithms~\cite{noauthor_cuda_nodate}.

Metropolis~\cite{andrieu2010particle} and Rejection resampling algorithms mitigate the issues of the prefix sum methods by avoiding the prefix sum entirely~\cite{Martino2010generalized}. Instead of computing a particle's offspring, each particle finds an ancestor to replace itself in both algorithms. The number of offspring for a particle is determined by the number of other particles that have selected \textit{it} to be their ancestor. To find a particle's ancestor, the Metropolis and Rejection resampling algorithms use the ratio between pairs of particle weights to iteratively search for particles with higher weighting to choose as ancestors.

The Rejection resampler is unbiased but requires an upper bound on the particle weights, and it can take a \textit{variable amount of time} to choose an ancestor for each particle~\cite{Martino2010generalized}. Varying execution times are undesirable as divergent code paths on GPUs have performance implications. We discuss this issue further in Section~\ref{section:gpuProgramming}. The Metropolis resampler is biased but does not require an upper bound on the weights. The time to find an ancestor for each particle is constant; consequently, \textit{Metropolis is a more desirable option for parallel resampling}. Unfortunately, both of these algorithms aiming to mitigate the numerical instability of prefix sum methods are shown to have worse execution time performance in comparisons with other methods ~\cite{murray_gpu_2012, murray2016parallel}. 

Both Rejection and 
Metropolis resamplers involve random memory accesses with significant performance penalties on modern, massively parallel computing hardware such as graphics processing units (GPUs). As we explore in Section \ref{section:coalesce}, the memory access patterns exhibited by an algorithm greatly affect GPU performance due to the memory architecture and the relatively slow memory access times compared to GPU processing speeds. Two techniques developed by Dülger et al. \cite{dulger_memory_2018}---Metropolis-C1 (C1) and Metropolis-C2 (C2)---aim to improve the  Metropolis resampling algorithm  performance on GPU platforms. These techniques recognise the problems with the memory access patterns generated by Metropolis and alter the access pattern of the original Metropolis algorithm to improve performance on GPUs but introduce a tuning parameter that adjusts both the speed and the quality of resampling.  As we discuss in Section~\ref{sec:metropolis_resampler}, the memory access patterns generated by Metropolis-C1 and Metropolis-C2 techniques still inhibit the performance of these methods.

In this article, we present a new algorithm, \textit{Megopolis}, to improve the performance of the Metropolis resampling algorithm by focusing on \textit{designing memory access patterns to exploit memory coalescing on modern GPUs to improve performance}. Importantly, we achieve performance improvements without introducing algorithm tuning parameters or affecting the resampling quality. Further, we prove that the convergence rate for Megopolis is the same as Metropolis and, consequently, achieve the same algorithmic complexity but with the significant performance improvements attained from coalesced memory access patterns. We compare our method to the original Metropolis algorithm and both C1 and C2 in terms of execution time and resampling quality on a GPU platform. We demonstrate that \textit{Megopolis} allows users to adopt the numerically stable Metropolis resampler in applications that demand speed without increasing bias and root-mean-square-error resampling. Further, to support the research community, we also \textit{open source} the complete project comprising a repository of source \footnote{see: \url{https://github.com/AdelaideAuto-IDLab/Megopolis}}.

The paper is organised as follows. Section~\ref{sec:background} provides preliminary background information, including particle filtering, GPU programming and memory access patterns.  Section~\ref{sec:metropolis_resampler} revisits the Metropolis resampler and its C1 and C2 variants. Section~\ref{sec:megopolis} presents our proposed Megopolis resampling algorithm. Section~\ref{sec:experimental_setup}  describes our experimental framework and performance evaluation measures. Section~\ref{sec:resultsDiscussion} presents our experimental results across the bias, mean squared error (MSE), and execution time. Section~\ref{sec:end_to_end_app} details an end-to-end application benchmark. Section~\ref{sec:conclusion} discusses concluding remarks.

\section{Background}~\label{sec:background}
In this section, we provide a brief overview of sequential important resampling (SIR) particle filters also known as bootstrap particle filter (BPF) since we will employ the SMC method in the end-to-end algorithm benchmark in the context of an estimation problem in Section~\ref{sec:end_to_end_app}. SIR filters are a common and easy to understand example of an SMC method for state-space estimation problems involving highly non-linear systems and noisy observations with non-Gaussian noise. 

Although we have used the SIR filter to illustrate the use of and benchmark the proposed resampling algorithm, other particle filter algorithms (\eg, Auxiliary Particle Filter (APF)~\cite{pitt1999filtering}, or Improved APF (IAPF)~\cite{elvira2019elucidating}) or Adaptive Importance Sampling (AIS) algorithms (\eg, Standard Population MC (PMC)~\cite{cappe2004population}, Mixture PMC (M-PMC)~\cite{cappe2008adaptive}, Deterministic Mixture PMC (DM-PMC)~\cite{elvira2017improving}, or Diverse PMC (D-PMC)~\cite{elvira2017population}),  employing resampling algorithms will also benefit from the performance benefits gained from massive parallelisation of the resampling procedure.

Notably, earlier efforts have focused on parallelisation of prediction and update processes of particle filters using techniques such as clustering techniques and dividing the population of particles into a set of sub-population particles~\cite{brun2002parallel,verge2015parallel} (so-called islands) as well as the implementation of distributed particle filters on specialised hardware such as FPGA~\cite{bolic2005resampling} or VLSI~\cite{hong2006design}. Similar to the approaches adopted in~\cite{dulger_memory_2018,nicely2019improved} , we adopted the capabilities of modern GPUs and programming paradigms to implement a SIR (or BPF) particle filter where each step---prediction, update and \textit{resampling}---is parallelised to benchmark the performance gains from the proposed resampling algorithm.  

Further, in the following, we provide an overview of the GPU programming model, coalesced memory access on modern GPUs and illustrate the problem of uncoalesced memory access patterns to help understand the proposed Megopolis algorithm. We begin with an introduction to the notations we have adopted.  

\subsection{Notations}
For notational clarity and simplicity, we use non-bold letters to denote scalar values (\eg, $x,w$), bold letters to denote vectors, \eg, $\boldsymbol{x} = [x^{(0)},\dots,x^{(N-1)}]^T$, and $\\$
$\boldsymbol{w} = [w^{(0)},\dots,w^{(N-1)}]^T$, while  $(\cdot)^T$ denotes the transpose of a vector. A uniform distribution of real numbers within the interval $[a,b)$ is denoted as $\mathcal{U}[a,b]$ while $\mathcal{U}\{a,b\}$ denotes a uniform distribution of integer numbers within the interval $[a,b]$. Further, $\mathcal{N}(\mu,\Sigma)$ denotes a Gaussian distribution with mean $\mu$ and co-variance $\Sigma$; $\mathbb{E}(\cdot)$ denotes the expectation operator. 

\subsection{SIR Particle filters}\label{sec:SIR} 
Particle filters (PF)~\cite{doucet2001introduction,gordon1997hybrid,gordon_novel_1993,ristic_beyond_2003} belong to a class of approximation methods for non-linear systems in the Bayesian filter family. The primary method of a particle filter is to use a random sampling process to approximate the probability distributions of interest~\cite{gordon_novel_1993}. Particle filters implement the random sampling process called the Monte Carlo (MC) method to approximate the belief density by a weighted set of independently and identically distributed (i.i.d) particles. 
The filtering algorithm involves three key processes: \textit{i)}~prediction--the particles are propagated using the system model; \textit{ii)}~update--the particle weights are updated based on noisy observations, and \textit{iii)}~\textit{resampling}---particles with low weights are removed. 

Formally, suppose that $x_{t-1}$ is the state of interest at time $t-1$, which generates an observation $z_{t-1}$ while $z_{1:t-1}$ denotes the measurements history up to time $t-1$. The belief density $p({x}_{t-1}|{z}_{1:t-1}) $ is approximated by a set of particles, $\{(w^{(i)}_{t-1},x^{(i)}_{t-1})\}_{i=0}^{N-1}$, 
where $N$ is the number of particles,  $w^{(i)}_{t-1}$ and $x^{(i)}_{t-1}$ are the weight and state of particle $i$ at time $t-1$ respectively with $\sum\limits_{i=0}^{N-1} w^{(i)}_{t-1} = 1$, and  $\delta(\cdot)$ denotes the Kronecker delta function. 

Each particle is predicted to time $t$ using a dynamic transition model expressed as:
\begin{align} \label{eq:prediction}
	{x}^{(i)}_{t|t-1} = f_{t-1}(x^{(i)}_{t-1},v_{t-1}),
\end{align}
where $f_{t-1}(\cdot,\cdot)$ is a dynamic transition function,  $v_{t-1}$ is the process noise, while the weight $w^{(i)}_{t|t-1} = w^{(i)}_{t-1}$ is maintained during the prediction step. 

In contrast, during the update step, the particle state is maintained, \ie, ${x}^{(i)}_{t} = {x}^{(i)}_{t|t-1}$, while its corresponding weight is updated as
\begin{align}\label{eq:update}
w^{(i)}_t = p(z_t|{x}^{(i)}_{t}) w^{(i)}_{t|t-1},
\end{align}
and then normalised to ensure that $\sum_{i=0}^{N-1}w^{(i)}_t = 1$ before the resampling stage. 

A typical problem with particles filter is particles depletion or degeneracy, \ie, weights are concentrated on a few particles, while the remaining particles have weights close to zero. After multiple update procedures, the reason is that the variance of weights increases and never decreases because the measurement likelihood function is often less scattered than the dynamic transition kernel~\cite{arulampalam2002a,doucet2000sequential,doucet2009tutorial}. 
A well-known method to prevent particle depletion is resampling
~\cite{douc2005comparison,hol2006resampling,sileshi_particle_2013}. The resampling procedure avoids particle degeneracy by pruning particles with small weights while duplicating particles with high weights. 

In this article, we employ the SIR particle filter, also known as the bootstrap resampling filter~\cite{gordon_novel_1993}, shown in Algorithm \ref{alg:sirpf} for an end-to-end application benchmark in Section~\ref{sec:end_to_end_app}. In this algorithm, $\boldsymbol{x}_t$ is the set of particles at time step $t$ before resampling; $\boldsymbol{\bar{x}}_t$ is the set of particles at time step $t$ after resampling.  The SIR particle filter algorithm can be separated into three stages. \textit{Stage 1}:  particles and particle weights are updated using the prediction model and the measurement likelihood model. \textit{Stage 2}: the particle weights are normalised. \textit{Stage 3}: the particles are resampled, generating a new set of particles. Notably, in line $9$, the resampling step approximates the belief density by a set of particles with the equal weights of $1/N$, i.e.,
\begin{align*} \label{eq:resampling_explanation}
    p(x_t|z_{1:t}) \approx \sum_{i=0}^{N-1} w^{(i)}_t \delta(x_t - x^{(i)}_t) \approx \dfrac{1}{N} \sum_{i=0}^{N-1} \delta(x_t - \bar{x}^{(i)}_t). 
\end{align*}

\begin{algorithm}[!tb]
	\footnotesize
	\caption{SIR Particle Filter (Bootstrap Particle Filter)}     \label{alg:sirpf}
	\begin{algorithmic}[1] 
		\Statex \textbf{Input}: $[\boldsymbol{\bar{x}}_{t-1}, z_t]$
		\Statex \textbf{Output}: $\boldsymbol{\bar{x}}_t$
		\LineCommentNoIdent {\texttt{Stage 1: Predict and Update}}
		\For {$i\gets 0$ to $N-1$} 
		    \State $x_t^{(i)} = f_{t-1}(\bar{x}^{(i)}_{t-1},v_{t-1})$ \Comment{prediction using \eqref{eq:prediction}}
		    \State $w^{(i)}_t = p(z_t | x_t^{(i)})$ \Comment{update using \eqref{eq:update} where $w^{(i)}_{t|t-1}=\dfrac{1}{N}$ is omitted.} 
		\EndFor
		\LineCommentNoIdent{ \texttt{Stage 2: Normalise Weights}}
		\State $r = \sum_{i=0}^{N-1} w^{(i)}_t$
		\For {$i\gets 0$ to $N-1$}
		\State $w^{(i)}_t = w^{(i)}_t / r$
		\EndFor
		\LineCommentNoIdent{ \texttt{Stage 3: Resample}}
		\State $\boldsymbol{\bar{x}}_t =$ RESAMPLE($[\boldsymbol{x}_t, \boldsymbol{w}_t]$)
	\end{algorithmic}
\end{algorithm}

\subsection{GPU Programming}
\label{section:gpuProgramming}

Implementing numerical algorithms efficiently on GPUs requires careful consideration of specific architectural features of graphics processing units. We briefly introduce GPU programming and code execution concepts to provide necessary insights into algorithm design decisions and their performance impact on GPUs. We provide a primer on threads, memory, and kernels in the context of GPU in the following.

On a GPU, instructions are executed in a single instruction, multiple threads (SIMT) architectures. In this architecture, threads are executed in groups of 32 parallel threads called \textit{warps}. Each thread in a warp initialises to the same program address but maintains its own instruction counter and register state. A warp will execute one instruction on all threads simultaneously until all threads have finished execution. If a thread's instruction counter differs from the warp instruction, that thread will be inactive for that cycle, reducing the efficiency gains from parallelisation. Consequently, diverging code paths should be avoided in GPU algorithm designs.

Executing code on a GPU generally involves threads fetching data from memory and executing instructions, as in a typical Von Neumann architecture. As expected, memory access times are orders of magnitude slower (involving hundreds of clock cycles) than instruction executions \cite{wong_demystifying_2010} and lead to performance bottlenecks. As a result, the manner in which threads in a warp access memory (their \textit{memory access pattern}) can greatly affect the execution time exhibited by threads in a warp. Consequently, efficient memory access patterns are critical to achieving performance improvements.

Kernels define the code to be executed on the GPU in parallel. The host machine must \textit{launch} a kernel onto the GPU. To launch a kernel, the host must specify the number of threads to launch with two parameters, $blockcount$ and $blocksize$. The total number of threads launched is $blockcount \times blocksize$. Two common ways of launching kernels are the monolithic and grid-stride loop launch. In a monolithic launch, the host launches one thread per datum. In a grid-stride loop launch, any number of threads can be launched, but each thread operates on multiple data. Launching kernels generates an added overhead as the GPU must prepare and schedule threads for execution.

\subsection{Memory access patterns}

\label{section:coalesce}
\begin{figure}[b]
    \begin{subfigure}[t]{\linewidth}
        \centering
        \includegraphics[scale=0.52]{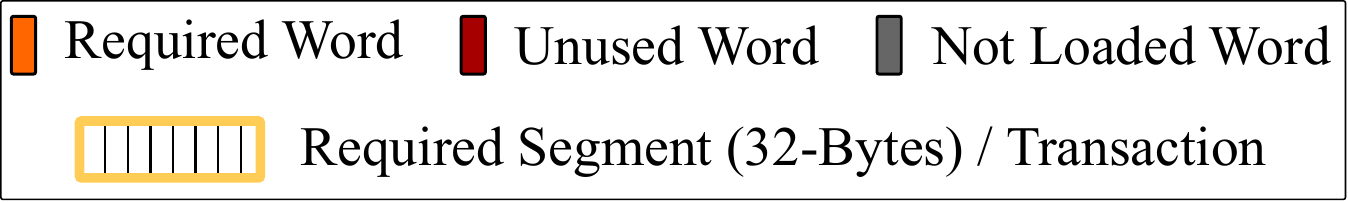}
        \includegraphics[scale=0.52]{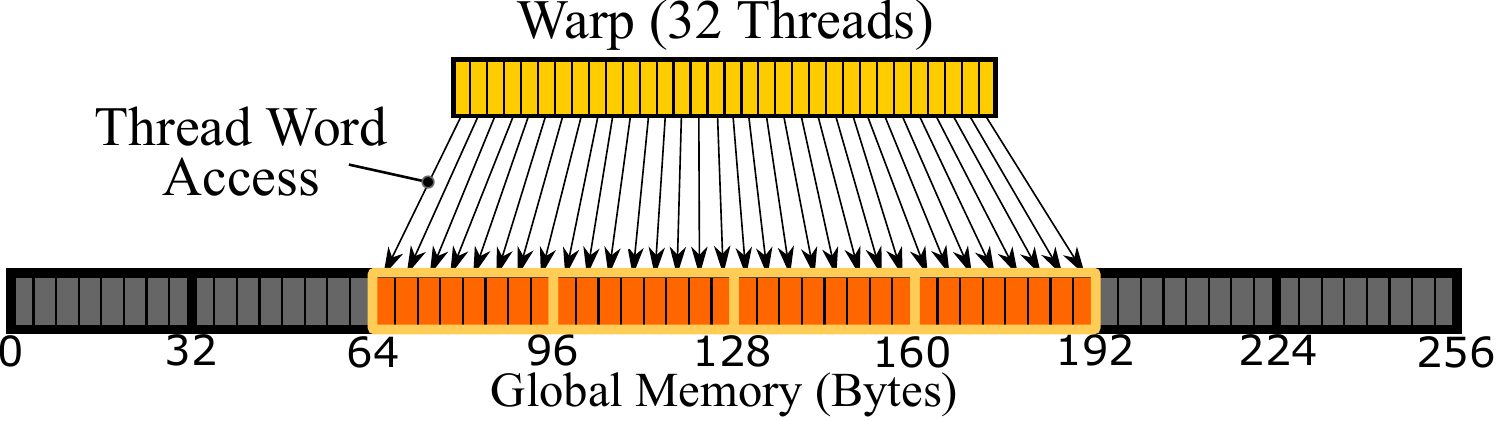}
        
        \caption{Simple access pattern. Requires 4 transactions.}
        \label{fig:simpleMemAccess}
    \end{subfigure}

    \begin{subfigure}[t]{\linewidth}
        \centering
        \includegraphics[scale=0.52]{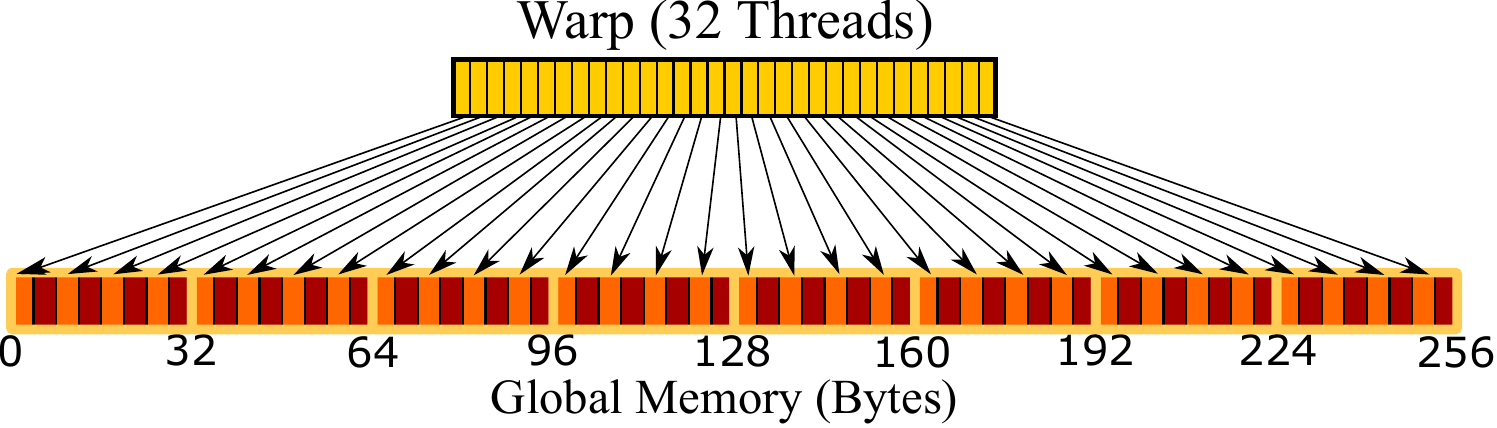}
        \caption{Strided access pattern with an offset of 2. Requires 8 transactions.}
        \label{fig:stridedMemAccess}
    \end{subfigure}

    \caption{An illustration of a warp's 32-byte memory transactions with different memory access patterns. Each thread in the warp points to the memory location of the 4-byte word it is accessing.}
    
    \label{fig:coalescedMem}
\end{figure}

Warps are responsible for servicing the memory accesses of all its owned threads. Warps \textit{transact} with memory in aligned $32$-byte segments, and if multiple threads request data from the same segment, then the warp only needs to fetch one segment to service those threads. Coalescing memory access combines many memory accesses into a single or as few memory transactions as possible. Consequently, \textit{coalesced memory access occurs when threads in a given warp access data that are physically close in the memory address space}~\cite{noauthor_cuda_nodate}. Memory accesses leading to coalesced memory access patterns can significantly improve memory throughput and reduce the execution time of numerical algorithms.

Certain memory access patterns can lead to uncoalesced memory access, and \textit{it is important to consider the resulting memory access patterns when designing an algorithm for GPU execution}. We provide a brief discussion on two contrasting access patterns to highlight the potential impact on performance as well as to aid the description of the proposed resampling algorithm, which can realise coalesced memory access patterns. For simplicity, we consider the scenario of a single warp where the threads are accessing 4-byte words from global memory, where global memory is only 256-bytes in size.

First, we look at the simplest access pattern that achieves high coalescence due to localised thread accesses. Here, thread $k \in \{1\dots32\}$ accesses the $k^{\rm th}$ word in a $32$-byte aligned array starting at byte $64$, illustrated in Fig.~\ref{fig:simpleMemAccess}. In this case, the warp only requires 4 transactions to service all 32 threads, and no word is unnecessarily loaded. 

Second, we look at a stridden access pattern that achieves poor coalescence due to separated thread accesses. Here, given some offset $o$, thread $k \in \{1\dots32\}$ accesses the $(ko)^{\text{th}}$ word in a $32$-byte aligned array starting at byte $0$, illustrated in Fig.~\ref{fig:stridedMemAccess}. With an offset of $o=2$, a warp requires 8 transactions to service all 32 threads, and every second word is unnecessarily loaded in this example. Notably, as $o$ increases, the thread accesses are further separated, and the impact of increased transactions are further exacerbated.

\section{Revisiting the Metropolis Resampler}~\label{sec:metropolis_resampler}

\begin{figure}[t]
    \begin{subfigure}[t]{\linewidth}
        \centering
        \includegraphics[scale=0.52]{figures/legend.pdf}
        \includegraphics[scale=0.52]{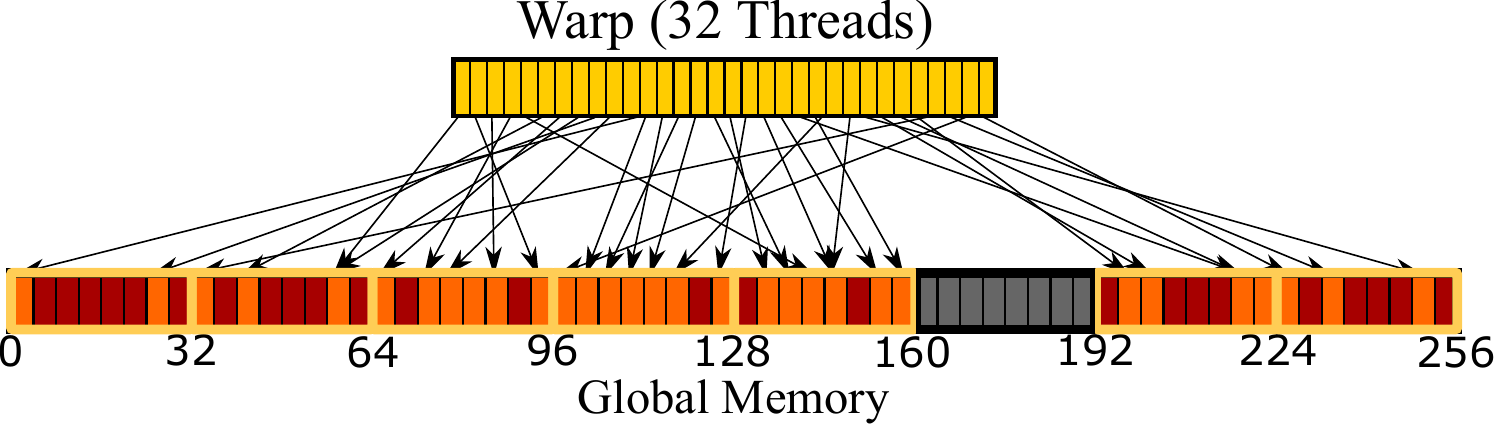}

    \end{subfigure}

    \caption{An illustration of the Metropolis random memory access pattern. The warp requires from 1 to 8 memory transactions when the section of memory is 256-bytes in size.}
    
    \label{fig:randomAccess}
\end{figure}

The Metropolis resampling algorithm designed for parallel execution avoids the common issue of numerical instabilities from computing the prefix sum of weights seen in other resampling methods. Although the algorithm has been analysed by other previous works~\cite{dulger_memory_2018,murray_gpu_2012,murray2016parallel}, we revisit Metropolis and its variants to explain practical and performance issues as well as to aid the explanation of our Megopolis algorithm.

For completeness, the Metropolis algorithm is described in Algorithm~\ref{alg:metroAlg}, and parameter definitions are explained in Table~\ref{tab:parameterList}. Metropolis operates by using the ratios between randomly selected particle weights to iteratively select particles with higher weights to duplicate. It generates a list containing an ancestor for each particle where the ancestor is the index of the particle that will replace the current particle. However, \textit{due to the random selection of particles, the algorithm manifests a random memory access pattern} illustrated in Fig.~\ref{fig:randomAccess}. In this example, thread $k \in \{1\dots32\}$ of a warp accesses the $i^{th}$ word where $i \sim \mathcal{U}\{0,63\}$. Notably, when the number of particles exceeds 256, it is possible for thread memory accesses to be completely separated where each thread in a warp accesses a weight from a unique segment for comparison. This can result in a maximum of 32 memory transactions for a single warp. Consequently, the algorithm implementations on GPUs leads to poor performance due to memory throughput bottlenecks.

Importantly, in the Metropolis algorithm, the convergence of the resampled distribution to the approximated posterior probability density function or the belief density of the system is determined by the number of iterations $B \in \mathbb{N}$ (an integer). Since $B$ is finite in practice, the algorithm always produces a biased sample--- the larger number of iterations leading to lower bias. Hence, $B$ needs to be selected with care, and the selection process of $B$ is described in detail in~\cite{murray2016parallel} and is given by
\begin{equation}
    B \geq \ceil[\Bigg]{\dfrac{\log(\epsilon)} {\log(1 - \dfrac{\mathbb{E}(\boldsymbol{w})}{w^{(p)}})} }
    \label{eq:bcalc}
\end{equation}
where $\ceil[\big]{\cdot}$ is the ceiling operator, $\epsilon$ is an error bound (maximum variation distance to the posterior probability density function) greater than 0; $w^{(p)} = \max(\boldsymbol{w})$ is the maximum of particle weights. In practice, we want to avoid calculating $B$ as it requires a weight summation to compute $\mathbb{E}(\boldsymbol{w})$ and an exhaustive search for $w^{(p)}$. Both the summation and exhaustive search increase execution time and, as discussed previously, performing a sum over the weights can lead to numerical instabilities. Instead, an appropriate value for $B$ can be chosen from some estimate of $\mathbb{E}(\boldsymbol{w}) / w^{(p)}$ either pre-computed experimentally or computed at runtime using a subset of the particles or another estimation method. 

\begin{table}[tb]
\caption{List of parameters used in Metropolis, C1, C2, and Megopolis.}
\label{tab:parameterList}
\centering
\renewcommand{\arraystretch}{1.1}
\begin{tabular}{|c|p{0.65\columnwidth}|}
\hline
\multicolumn{2}{|c|}{\textbf{Common Parameters}} \\
    \hline
    $\boldsymbol{x}_t$ & Particles at time \textit{t} \\
    $\boldsymbol{w}_t$ & Weights at time \textit{t} \\
    $\boldsymbol{\bar{x}}_t$ & Resampled particles at time \textit{t} \\
    $N$ & Number of particles \\
    $B$ & Number of iterations \\
    $i$ & Particle index \\
    $k$ & Ancestor index \\
    $u$ & Uniform random number between 0 and 1 for ancestor selection \\
    $b$ & Inner loop iteration index \\ 
    $j$ & Selected particle index for comparison \\
    \hline
    \multicolumn{2}{|c|}{\textbf{Metropolis-C1/C2 Specific Parameters}} \\ 
    \hline
    $P_{\textrm{size}}$ & Number of bytes in a chosen partition of $\boldsymbol{w}_t$ \\
    $N_{\textrm{part}}$ & Number of $P_{\textrm{size}}$  partitions in $\boldsymbol{w}_t$\\
    $N_{w}$ & Number of single-precision floating-point weights in a partition 
    \\
    $i_{\text{warp}}$ &  Warp index \\
    $p$ &  Selected partition \\
    \hline
    \multicolumn{2}{|c|}{\textbf{Megopolis Specific Parameters}} \\
    \hline
    $\boldsymbol{o}$ & List of random integer offsets \\ 
    $i_{aligned}$ & Starting index of initial aligned segment \\ 
    $o_{aligned}$ & Segment aligned offset \\
    $o_{unaligned}$ & Unaligned offset within segment \\
    \hline
\end{tabular}
\end{table}
\begin{algorithm}[!t]
	\footnotesize
	\caption{Metropolis Resample} \label{alg:metroAlg}
	\begin{algorithmic}[1] 
		\Statex \textbf{Input}: $[\boldsymbol{x}_t, \boldsymbol{w}_t]$
		\Statex \textbf{Output}: $\boldsymbol{\bar{x}}_t$
		\For {$i\gets 0$ to $N-1$} 
		    \State $k \gets i$
		    \For {$b\gets 0$ to $B-1$} 
    		    \State $u \sim \mathcal{U}[0,1]$
    		    \State $j \sim \mathcal{U}\{0, N-1\}$ \Comment{select a particle index for comparison}
    		    \If{$u \leq w_t^{(j)} / w_t^{(k)}$} 
    		        \State $k \gets j$
    		    \EndIf
		\EndFor
		\State $\bar{x}_t^{(i)} = x_t^{(k)}$
	\EndFor
	\end{algorithmic}
\end{algorithm}

The Metropolis resampler can be implemented with a single GPU kernel. The outer loop (lines 1 to 11 in Algorithm~\ref{alg:metroAlg}) can be executed in parallel such that each thread is assigned one particle to find the ancestor for. The particle weights are stored in memory; hence, each comparison requires memory access from a \textit{random location in memory} resulting in the random access pattern discussed earlier. \textit{The generated memory access becomes an issue due to the inability of the GPU to coalesce the random memory accesses as the number of particles increases.}

\subsection{Metropolis-C1 and Metropolis-C2}
\begin{figure}[b]
    \centering
    \includegraphics[scale=0.52]{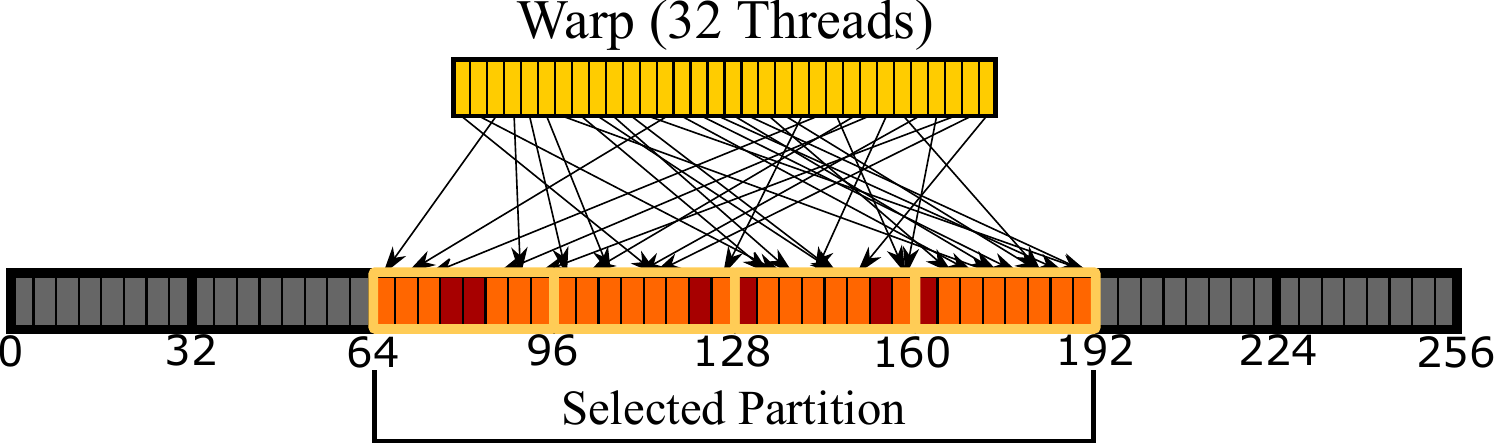}
    \includegraphics[scale=0.52]{figures/legend.pdf}
    \caption{An illustration of the partitioned random memory access pattern used by the C1 and C2 algorithms when the selected partition starts at the 64th byte and is 128 bytes long.}
    \label{fig:segementRandomAccess}
\end{figure}

\begin{algorithm}[!t]
	\footnotesize
	 \caption{Metropolis-C1 Resample} \label{alg:metroC1Alg}
	\begin{algorithmic}[1] 
		\Statex \textbf{Input}: $[\boldsymbol{x}_t, \boldsymbol{w}_t]$
		\Statex \textbf{Output}: $\boldsymbol{\bar{x}}_t$
		\State $N_{\textrm{part}} \gets 4 N / P_{\textrm{size}}$ \Comment{number of partitions} 
        \State $N_{w} \gets P_{\textrm{size}} / 4$ \Comment{number of weights in a partition} 
		\For {$i\gets 0$ to $N-1$} 
		    \State $k \gets i$
		    \State $ i_{\text{warp}}  \gets \lfloor i / 32 \rfloor$ \Comment{compute warp index}
		    \State $p \sim \mathcal{U}\{0,N_{\textrm{part}}-1\} \mid i_{\text{warp}} $ \LineComment{select common partition for threads  in a warp}
		    \For {$b\gets 0$ to $B-1$} 
    		    \State $u \sim \mathcal{U}[0,1]$
    		    \State $j \sim \mathcal{U}\{p \cdot N_{w}, (p + 1) \cdot N_{w} - 1\}$ \LineComment{select a particle index for comparison}
    		    \If{$u \leq w_t^{(j)} / w_t^{(k)}$} 
    		        \State $k \gets j$
    		    \EndIf
		\EndFor
		\State $\bar{x}_t^{(i)} = x_t^{(k)}$ 
	\EndFor
	\end{algorithmic}
\end{algorithm}

\begin{algorithm}[!t]
	\footnotesize
	 \caption{Metropolis-C2 Resample}     \label{alg:metroC2Alg}
	\begin{algorithmic}[1] 
		\Statex \textbf{Input}: $[\boldsymbol{x}_t, \boldsymbol{w}_t]$
		\Statex \textbf{Output}: $\boldsymbol{\bar{x}}_t$
		\State $N_{\textrm{part}} \gets 4 N / P_{\textrm{size}}$ \Comment{number of partitions} 
        \State $N_{w} \gets P_{\textrm{size}} / 4$ \Comment{number of weights in a partitions} 
		\For {$i\gets 0$ to $N-1$} 
		    \State $k \gets i$
		    \State $i_{\text{warp}} \gets \lfloor i / 32 \rfloor$ \Comment{compute warp index}
		    \For {$b\gets 0$ to $B-1$} 
    		    \State $u \sim \mathcal{U}[0,1]$
    		    \State $p \sim \mathcal{U}\{0,N_{\textrm{part}}-1\} \mid i_{\text{warp}} $ \LineComment{select common partition for threads in a warp}
    		    \State $j \sim \mathcal{U}\{p \cdot N_{w},(p + 1) \cdot N_{w} - 1\}$ \LineComment{select a particle index for comparison}
    		    \If{$u \leq w_t^{(j)} / w_t^{(k)}$} 
    		        \State $k \gets j$
    		    \EndIf
		\EndFor
		\State $\bar{x}_t^{(i)} = x_t^{(k)}$
	\EndFor
	\end{algorithmic}
\end{algorithm}

The Metropolis-C1 (simply C1 henceforth) and Metropolis-C2 (C2 henceforth) algorithms are an adaptation of the Metropolis algorithm aiming to improve memory coalescence~\cite{dulger_memory_2018}. We revisit C1 and C2 in Algorithms~\ref{alg:metroC1Alg} and~\ref{alg:metroC2Alg}, respectively. The additional parameters introduced in C1 and C2 are described in Table~\ref{tab:parameterList}. 

Both C1 and C2 algorithms divide the particles into partitions; each partition contains a user-defined number of bytes. Conceptually, the algorithms attempt to force all threads in a given warp to only access memory from one partition at a time (line $6$ in Algorithm~\ref{alg:metroC1Alg} and line $8$ in Algorithm~\ref{alg:metroC2Alg}) to increase the localisation of thread memory accesses. Consequently, the number of global memory transactions required by each thread warp is related to the partition size $P_{\textrm{size}}$ instead of random memory accesses---illustrated in Fig.~\ref{fig:randomAccess}---resulting from the execution of the Metropolis algorithm. The realised partitioned random memory access pattern generated by C1 and C2 is illustrated in Fig.~\ref{fig:segementRandomAccess}. 

The partition approach adopted by C1 and C2 forces a \textit{compromise between speed and bias}. In C1, Algorithm~\ref{alg:metroC1Alg}, each warp selects one partition outside the inner loop (line $6$) and performs Metropolis resampling using only the data within that partition.  Reducing the partition size reduces the number of memory transactions required but restricts the range of data each thread can view when choosing an ancestor. The C2 algorithm---Algorithm~\ref{alg:metroC2Alg}---reduces bias by randomly selecting a new partition at each iteration of the inner loop (line $8$). However, this alteration affects performance because the adopted method forces each thread to generate an extra random number for each algorithm iteration. 

\section{Megopolis}~\label{sec:megopolis}
We propose a variation of the Metropolis algorithm, referred to as the Megopolis algorithm, to improve performance on massively parallel hardware architectures of modern graphics processing units. We aim to improve memory coalescence and obviate the need for a tuning parameter, such as the partition size in C1 and C2, whilst not sacrificing the quality of the results measured by \textit{bias} and \textit{mean squared error}---defined in Section~\ref{sec:experimental_setup}. 

\begin{algorithm}[!t]
	\footnotesize
	\caption{Megopolis Resample}     \label{alg:megoAlg}
	\begin{algorithmic}[1] 
		\Statex \textbf{Input}: $[\boldsymbol{x}_t, \boldsymbol{w}_t]$
		\Statex \textbf{Output}: $\bar{\boldsymbol{x}}_t$
		\State $\boldsymbol{o} \gets [o^{(0)},\dots,o^{(B-1)}]^T$
		
		\For{$b\gets 0$ to $B-1$}
            \State $o^{(b)} \sim \mathcal{U}\{0, N-1\}$
        \EndFor
		\For {$i\gets 0$ to $N-1$} 
		    \State $k \gets i$
		    \State $i_{aligned} \gets i - (i \mod 32)$ 		    \Comment{compute aligned index}
		    \For {$b\gets 0$ to $B-1$} 
    		    \State $o_{aligned} \gets o^{(b)} - (o^{(b)} \mod 32)$    		    \Comment{compute aligned offset}
    		    \State $o_{unaligned} \gets (i + o^{(b)}) \mod 32$     		    \Comment{compute unaligned offset}
    		    \State $j \gets (i_{aligned} + o_{aligned} + o_{unaligned}) \mod N$     		    \LineComment{select a particle index for comparison}
    		    \State $u \sim \mathcal{U}[0,1]$
    		    \If{$u \leq w_t^{(j)} / w_t^{(k)}$} 
    		        \State $k \gets j$
    		    \EndIf
		\EndFor
		\State $\bar{x}_t^{(i)} = x_t^{(k)}$
	\EndFor
	\end{algorithmic}
\end{algorithm}

\begin{figure}[t]
    \begin{subfigure}[t]{\linewidth}
        \centering
        \includegraphics[scale=0.52]{figures/legend.pdf}
        \includegraphics[scale=0.52]{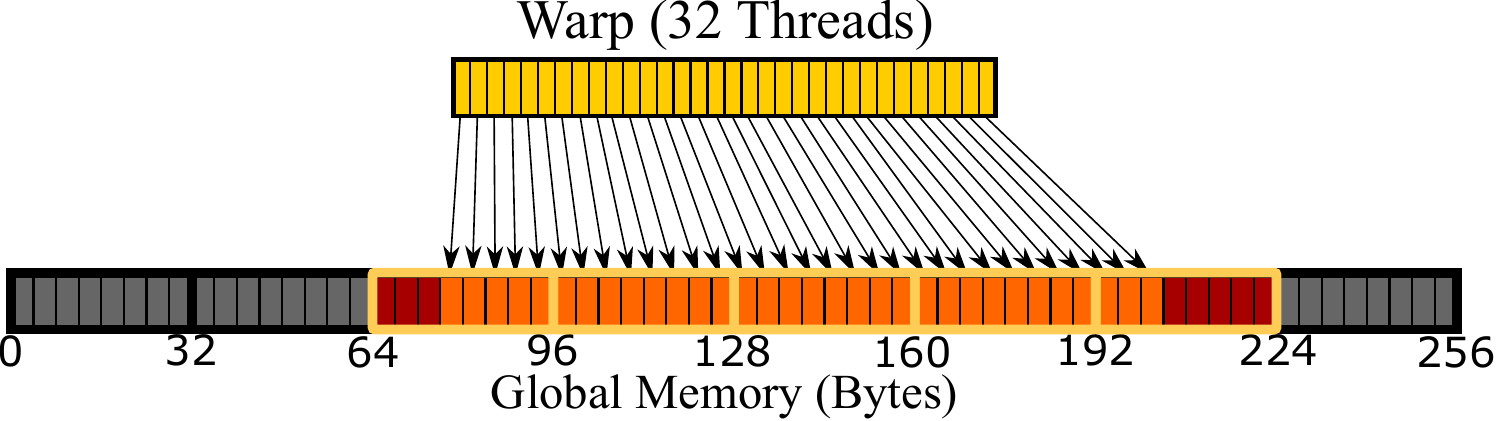}
        \caption{Misaligned but sequential memory access pattern. Requires at most 5 transactions.}
        \label{fig:misalignedAccess}
    \end{subfigure}
    \begin{subfigure}[t]{\linewidth}
        \centering
        \includegraphics[scale=0.52]{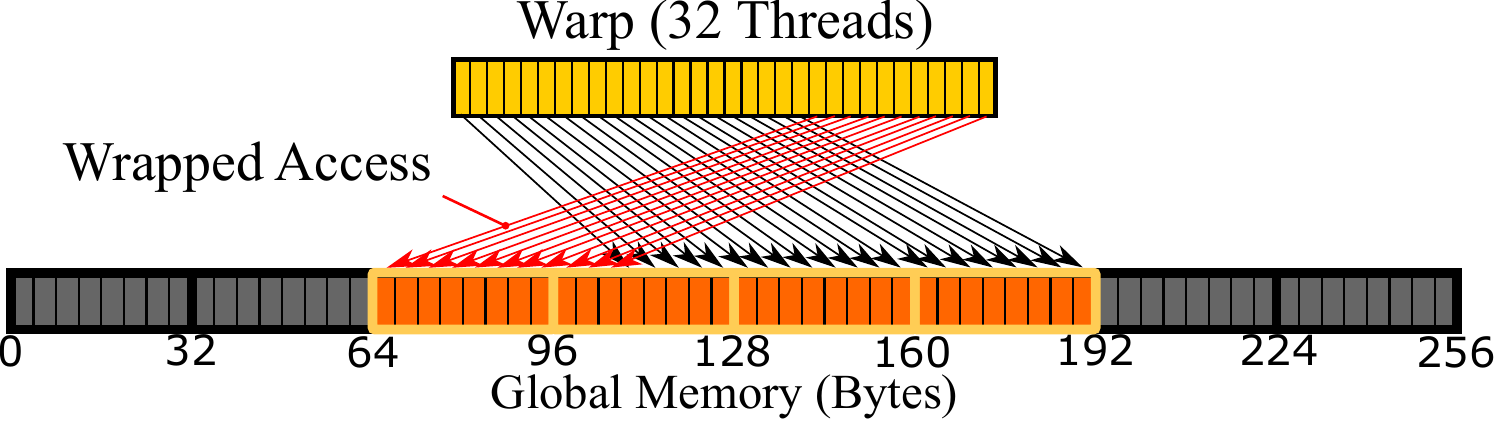}
        \caption{Wrapped sequential memory access pattern. Requires 4 transactions}
        \label{fig:wrappedAccess}
    \end{subfigure}
    \caption{An illustration of the memory access patterns exhibited by the Megopolis algorithm.}
    \label{fig:megopolisAccess}
\end{figure}

Algorithm~\ref{alg:megoAlg} describes Megopolis algorithm. The key concept underlying the construction of the Megopolis resampler is to avoid the need for individual threads to generate random indices for weight comparison and, thereby, remove the consequential generation of the undesirable random memory access patterns to global memory illustrated in Fig.~\ref{fig:randomAccess}. In contrast to previous algorithms, we compute $B$ random integers (lines 1--4) for use as random offset values, necessary for the $B$ comparisons the algorithm must perform (lines 8--16), as a separate task before the resampling process. Subsequently, during parallel execution, each random integer is used as an offset for all threads, simultaneously in the inner loop, to choose the next weight compared to the currently selected ancestor.

The random integer offsets are uniformly sampled from $[0, N-1]$ to ensure that every particle has a uniform likelihood of choosing any other particle for comparison. Here, $\boldsymbol{o}$ is a set of random offsets uniformly sampled from $[0, N-1]$; in contrast to previous algorithms, $j$ now uses the random offsets $o^{(b)} \in \boldsymbol{o}$ instead of generating another new random number to select particles for comparison. Given a set of random integers $\boldsymbol{o} = [o^{(0)},\dots,o^{(B-1)}]^T$, at iteration $i$ of the main loop and iteration $b$ of the inner loop, the idea is to use the value $o^{(b)}$ as an offset to the index of the $i^{\text{th}}$ particle to find the $b^{\text{th}}$ particle index for comparison. This can be achieved by adding the offset to the $i^{\text{th}}$ index and using modulo arithmetic to ensure the resulting index $j \in [0, N-1]$, i.e., $(i + o^{(b)}) \bmod N$. In the context of parallel execution, where the main loop is performed in parallel, this offset formulation would result in every thread in a warp accessing memory directly next to the memory accessed by its neighbours. This would achieve a localised \textit{sequential but misaligned} access pattern as illustrated in Fig.~\ref{fig:misalignedAccess}. However, Fig.~\ref{fig:misalignedAccess} demonstrates that this access pattern can result in an additional unnecessary memory transaction, where a warp of 32 threads requires at most 5 transactions with potentially 8 unnecessary word loads.

To remove the aforementioned unnecessary memory transaction, we force all memory accesses of threads in a warp to be within a memory aligned partition. We choose a partition for each warp, namely the warp partition, to be large enough. Each thread within that warp can access exactly one unique particle, i.e., for a warp with 32 threads, we use a partition size fits exactly 32 particle weights. To force memory accesses to be within this partition, memory accesses that extend past the end of the warp partition are \textit{wrapped} around to the start of the partition. In contrast to the sequential but misaligned formulation, at iteration $i$ of the main loop and iteration $b$ of the inner loop, the particle index for comparison is computed by first finding the starting index of the warp partition and then computing the offset within that partition. The starting index is found by offsetting the aligned index of the $i^{\text{th}}$ particle with the aligned offset from $o^{(b)}$. These aligned components are computed by rounding down to the nearest multiple of the partition size (lines 7 and 9). The offset within the partition is computed by wrapping the sum of $i$ and $o^{(b)}$ to be within the partition size, i.e., $i + o^{(b)} \mod 32$ for a warp with 32 threads (line 10). This achieves a highly localised \textit{wrapped sequential} access pattern illustrated in Fig.~\ref{fig:wrappedAccess}.

When using the wrapped sequential access pattern during parallel execution, all threads within a warp will access memory within an aligned segment of memory while each unique offset still produces a unique index. In the example presented in Fig.~\ref{fig:wrappedAccess}, this formulation results in a memory access pattern where a warp of 32 threads requires exactly four memory transactions with zero unnecessary word loads. This access pattern closely follows the highly efficient sequential memory access pattern illustrated in Fig.~\ref{fig:simpleMemAccess}.

Notably, the proposed Megopolis resampling algorithm requires the small but additional generation and storage of $B$ integers and, therefore, an additional runtime and memory complexity of $\mathcal{O}(B)$ compared to the Metropolis algorithm and its C1 and C2 variants. The generation of these random integers can be performed in parallel to reduce the runtime complexity to $\mathcal{O}(1)$.

\subsection{Selecting Number of Iterations}\label{sec:mego-proof}
In the following, we provide proof to confirm the proposition that the number of iterations $B$ for Metropolis and the proposed Megopolis can indeed be described by equation~\eqref{eq:bcalc} and consequently, bear the same rate of convergence of the resampled particle distribution to the estimated belief density of the system. This is an important result that ensures the number of iterations, hence the time-complexity, to achieve a desirable resampling quality is the same as that of Metropolis.

\begin{proposition}\label{prop:B_of_megopolis} The rate of convergence of Megopolis is the same as Metropolis and is described by the minimum number of iterations $B$ to achieve an error $\epsilon \in (0,1]$ as the maximum variation distance to the posterior probability density function $p(x_t|z_{1:t})$, where:
\begin{align}\label{eq:proposition}
    B = \ceil[\Bigg]{\dfrac{\log(\epsilon)} {\log(1 - \dfrac{\mathbb{E}(\boldsymbol{w})}{w^{(p)}})} }
\end{align}
and $w^{(p)} = \max(\boldsymbol{w})$ is the maximum weight of particles. 

\end{proposition}
\begin{proof}
Since the rate of convergence is described by the number of iterations $B$, 
we aim to show that the number of iterations $B$ for the Megopolis algorithm is the same as for the Metropolis algorithm. As explained in~\cite{murray2016parallel}, the convergence rate to the target probability distribution $p(x_t|z_{1:t})$ 
is predominantly dependent on the particle with the highest weight. If we define $x^{(p)}$ as the particle with the maximum weight $w^{(p)}$ 
then we want to estimate the value of $B$ which satisfies the probability that the particle $x^{(i)},~i \in \{0,\dots,N-1\}$ will choose $x^{(p)}$ as its ancestor given $B$ and $\boldsymbol{w}$, with a desired error $\epsilon \in (0,1]$,~\ie,
\begin{equation} \label{eq:B_condition}
    \textrm{Pr}(x^{(i)} = x^{(p)} | B, \bm w) \geq \dfrac{w^{(p)}}{\sum_{i=0}^{N-1} w^{(i)} } - \epsilon. 
\end{equation}

For compactness, let $P_{B} = \textrm{Pr}(x^{(i)} = x^{(p)} | B, \bm w)$. We investigate how $P_{B}$ changes over the iterations, by considering two aspects: i) the probability of selecting $x^{(p)}$ for the particles that have \textit{not} selected $x^{(p)}$ at iteration $B-1$, ii) the probability of not selecting $x^{(p)}$ for particles that have \textit{already selected} $x^{(p)}$ at iteration $B-1$. Note that since all particles are given the same offset value $o^{(b)}$ at each iteration, the probability of selecting any particle index in line $9$ of Algorithm~\ref{alg:megoAlg} is $\dfrac{1}{N}$.

\begin{itemize}
    \item[i)] Since the ratio $w^{(p)} / w^{(j)} \geq 1 ~\forall j \in {0,\dots,N-1}$, the probability of selecting $x^{(p)}$ at iteration $B$ is the probability that a particle that has not already selected $x^{(p)}$ at iteration $B-1$ chooses $x^{(p)}$ for comparison at iteration $B$, \ie, 
    \begin{align}
        (1-P_{B-1}) \cdot \dfrac{1}{N}
    \end{align}
    \item[ii)] The probability of deselecting $x^{(p)}$ at iteration $B$ for particles that have \textit{already selected} $x^{(p)}$ at iteration $B-1$ is 
    \begin{align}
        \sum_{i=0,i\not =p}^{N-1}\dfrac{w^{(i)}}{w^{(p)}} \cdot \dfrac{1}{N} \cdot P_{B-1}= \big[\dfrac{\mathbb{E}(\boldsymbol{w})}{w^{(p)}} - \dfrac{1}{N} \big] P_{B-1}
    \end{align}
\end{itemize}

Thus, at iteration $B$, $P_B$ can be computed as followed:
\begin{align} 
    P_B &= P_{B-1} + (1-P_{B-1}) \cdot \dfrac{1}{N} - \big[\dfrac{\mathbb{E}(\boldsymbol{w})}{w^{(p)}} - \dfrac{1}{N} \big] P_{B-1} \notag \\
    &= \dfrac{1}{N} + P_{B-1} (1 - \dfrac{\mathbb{E}(\boldsymbol{w})}{w^{(p)}} ). \label{eq:pkp1}
\end{align}

For a given $P_0$, we can write \eqref{eq:pkp1} in terms of $B$ as follows:

\begin{equation}\label{eq:P_{B-1}p1}
    P_{B} = (1 - \dfrac{\mathbb{E}(\boldsymbol{w})}{w^{(p)}} )^{B} P_0 + \dfrac{1}{N} \sum_{i=0}^{B-1} (1 - \dfrac{\mathbb{E}(\boldsymbol{w})}{w^{(p)}})^i
\end{equation}

By setting $P_0 = 0$ and  

$P_B = w^{(p)} / [\sum_{i=0}^{N-1} w^{(i)}] \times (1 - \epsilon) = (1-\epsilon)/(N \dfrac{\mathbb{E}(\boldsymbol{w})}{w^{(p)}})$ which satisfies \eqref{eq:B_condition}, we have:

\begin{align}
    &\dfrac{(1-\epsilon)}{N \dfrac{\mathbb{E}(\boldsymbol{w})}{w^{(p)}}} = \dfrac{1}{N} \cdot \dfrac{(1-\dfrac{\mathbb{E}(\boldsymbol{w})}{w^{(p)}})^B - 1}{(1-\dfrac{\mathbb{E}(\boldsymbol{w})}{w^{(p)}}) -1}  \\
    &\Leftrightarrow (1-\dfrac{\mathbb{E}(\boldsymbol{w})}{w^{(p)}})^B  = \epsilon. 
\end{align}

Thus, it suffices to select $B = \ceil[\Bigg]{\dfrac{\log(\epsilon)} {\log(1 - \dfrac{\mathbb{E}(\boldsymbol{w})}{w^{(p)}})} }$ as the minimum number of iterations to satisfy condition  \eqref{eq:B_condition}. Hence the minimum number of iterations $B$ described by \eqref{eq:bcalc} for the Metropolis algorithm is identical to that for the proposed Megopolis algorithm. Consequently, the rate of convergence of Megopolis is the same as Metropolis.

\end{proof}
\section{Experiments}~\label{sec:experimental_setup}
We detail the settings employed in the extensive numerical experimental regime and introduce the performance measures adopted to assess the performance of resampling algorithms. 

\vspace{2mm}
\noindent\textbf{Numerical Experimental Settings.~}To obtain the weight sequences $\bm w$ for the experiments, we use two weight-generation methods. 

For the first method, we follow the weight generation method proposed in~\cite{murray2016parallel}. We use this method to evaluate the performance of the algorithms as the particle weighting becomes concentrated on a few particles, a typical scenario demanding a resampling process. The weights are generated as follows: 

\begin{equation}
    w^{(i)} = \dfrac{1}{\sqrt{2 \pi}} \exp{\big(\dfrac{-1}{2} (x^{(i)} - y)^2\big)},
    \label{eq:distribution}
\end{equation}
where $x^{(i)} \sim \mathcal{N}(0,1)$ (drawn from a zero-mean Gaussian distribution with a co-variance $\Sigma=1$) and $y$ controls the weight distribution. As $y$ increases, the coefficient of variation (CV) $\sigma / \mu$ increases where $\sigma$ and $\mu$ are the standard deviation and mean of the weights, respectively. Consequently, a larger CV corresponds to few particles holding the majority of the weight. We generate weights using 5 different values for $y$ as illustrated in the distributions in Fig.~\ref{fig:murrayDist}. From Fig.~\ref{fig:murrayDist}, we can observe that as $y$ increases, most particles will have weights in the 0 and 1 weight bins, with few particles in the higher weight bins, simulating the particle degeneracy problem. 

For the second method, we adopt a similar approach in~\cite{dulger_memory_2018} used for evaluating resampling algorithms by sampling from the gamma distribution. We use the gamma distribution to evaluate the algorithms on a rich variety of particle weight distributions. The gamma distribution is defined by the following probability density function: 
\begin{equation}
    p(x \mid \alpha, \beta)  = \dfrac{\beta^{\alpha} x^{\alpha - 1} e^{-\beta x}}{\Gamma(\alpha)} \forall x, \alpha, \beta > 0,
    \label{eq:gammaEq}
\end{equation}
where $\alpha$ is the shape parameter for the gamma distribution $\Gamma(\alpha)$ and $\beta$ is the rate parameters of the gamma distribution. The shape and rate can be adjusted to obtain different distributions. We generate 5 different distributions by setting $\beta$ to 1 and selecting 5 values for $\alpha$: $0.5, 2.0, 3.0, 10.0, \text{and~} 50.0$. The resulting weight distributions are shown in Fig. \ref{fig:gammaDist}.

\begin{figure}[h]
    \centering
    \begin{subfigure}[t]{\linewidth}
        \centering
        \includegraphics[width=0.8\linewidth]{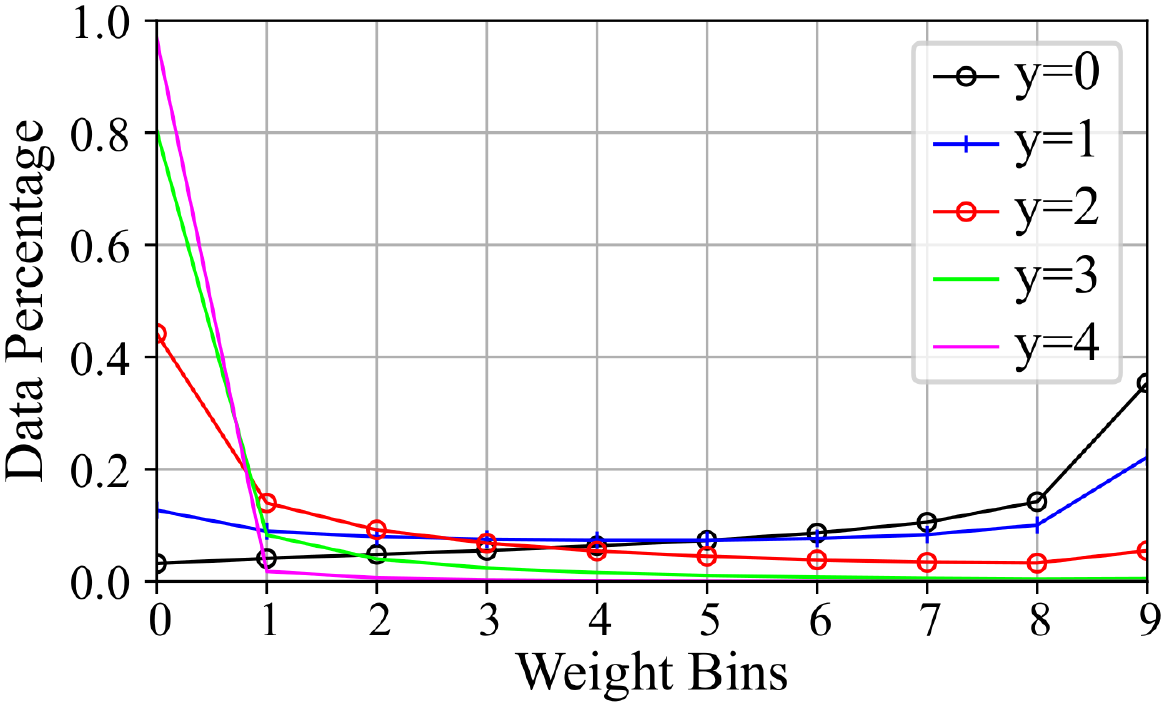}
        \caption{Weight distribution of a weight sequence generated from \eqref{eq:distribution}.}
        \label{fig:murrayDist}
    \end{subfigure}
    \par\bigskip 
    \begin{subfigure}[t]{\linewidth}
        \centering
        \includegraphics[width=0.8\linewidth]{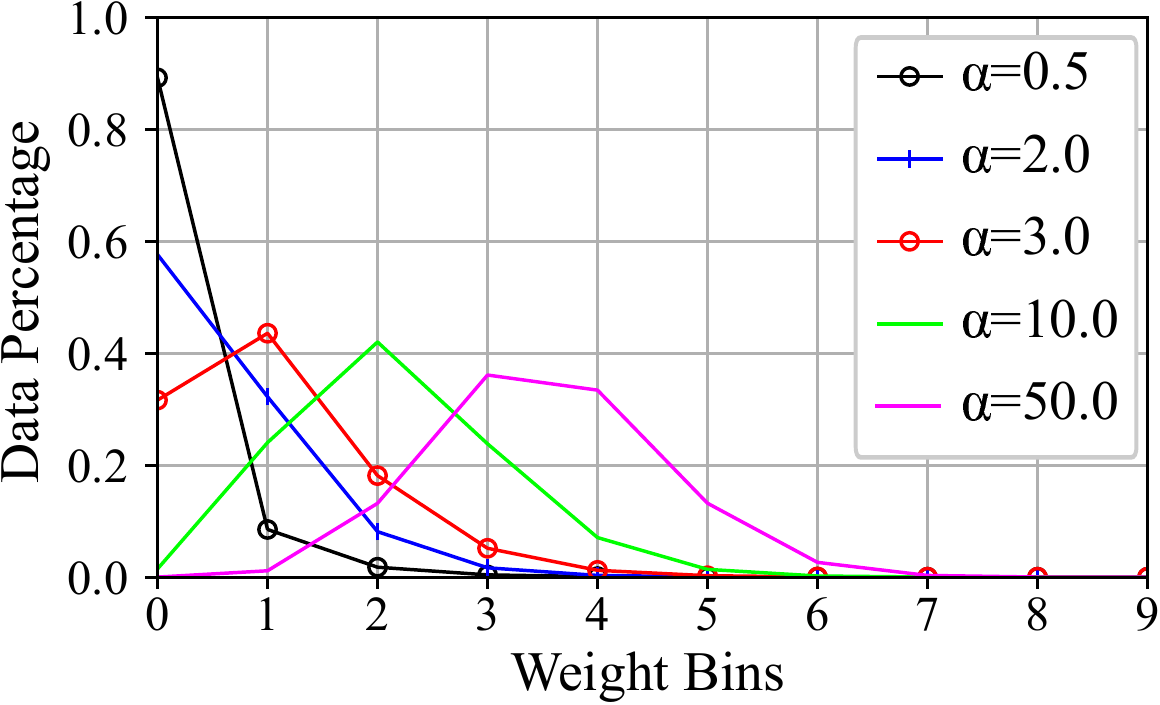}
        \caption{Weight distribution of a weight sequence generated from the gamma distribution.}
        \label{fig:gammaDist}
    \end{subfigure}
    
    \caption{Illustrations of the distribution of weights in weight sequences of size $N = 4,194,304$, generated using  \eqref{eq:distribution} and by sampling from the gamma distribution. The X-axis represents the bins of the weights, and the Y-axis represents the percentage of particles with weights in a given bin.}
    \label{fig:weightGen}
\end{figure}

\vspace{2mm}
\noindent\textbf{Algorithm Implementation Details.~}We implement each algorithm using the Compute Unified Device Architecture (CUDA) C++ platform on a Tesla K40m GPU. All resampling GPU kernels are implemented with a grid-stride loop to accommodate a varying number of particles easily. With the exception of the improved systematic kernels, they are launched with a 512 $blockcount$ and a 256 $blocksize$. The improved systematic kernels are launched with a 512 $blockcount$ and a 64 $blocksize$ to utilise the shared memory optimisation~\cite{nicely2019improved}. When a prefix sum is required for resampling, we employ the efficient parallel prefix sum algorithms provided by the \texttt{Thrust}  library~\cite{noauthor_cuda_nodate}. We use \texttt{XORWOW PRNG} pseudo-random number generator (PRNG) in the \texttt{CURAND} library for all random number generations in the kernels. The PRNGs are only initialised once outside of the resampling kernels, and the states are saved in global memory. For \textit{all} algorithms, when we are required to generate random numbers, the PRNG states are loaded from memory in a coalesced way and used to generate the random numbers in parallel. When a kernel has finished generating random numbers, the PRNG state is stored back into global memory. The initialisation time of the PRNGs is excluded from execution time; however, the loading and storing of the PRNG states are included.  

\subsection{Performance Evaluation Measures}
Resampling algorithms are often assessed empirically for mean squared error (MSE), bias, and execution time \cite{murray2016parallel, dulger_memory_2018}. 

\vspace{2mm}
\noindent\textbf{Mean Square Error (MSE).~}To determine the squared error (SE) we compare the offspring vector $\bm o_k$, where $o_k^{(i)}$ is the number of offspring of the $i^{\text{th}}$ particle, to the expected offspring computed from the weight vector $\bm w$ as described in~\cite{murray2016parallel} using:
\begin{equation}
\text{SE}(\bm o_k) = \sum_{i=0}^{N-1}\big(o_k^{(i)} - \dfrac{Nw^{(i)}} {\sum_{j=0}^{N-1}w^{(j)}}\big)^2.
\end{equation}

The MSE is then calculated from $K$ offspring vectors {$\bm o_{1},...,\bm o_{K}$} by taking the sample mean of squared errors:
\begin{equation}
\text{MSE}(\bm o) = \dfrac{1}{K} \sum_{k=1}^{K} \text{SE}(\bm o_k).
\end{equation}

\vspace{2mm}
\noindent\textbf{Bias.~}Notably, the MSE can be rewritten as a sum of two separate components, bias and variance highlighted in~\cite{murray2016parallel}:
\begin{equation}
\text{MSE}(\bm o) = \text{Var}(\bm o) + ||\text{Bias}(\bm o)||^2.
\end{equation}
Here, 
\begin{align}
    \text{Var}(\bm o) &= \sum_{i=0}^{N-1} \text{Var}(o^{(i)}), \\
    ||\text{Bias}(\bm o)||^2 &= \sum_{i=0}^{N-1}\big(\hat o^{(i)} - \dfrac{N \hat{w}^{(i)}}{\sum_{j=0}^{N-1}w^{(j)}}\big)^2
\end{align}
where,

\begin{align}
    \hat o^{(i)} &= \dfrac{1}{K} \sum_{k=0}^{K-1} o_k^{(i)}, \\
    \text{Var}(o^{(i)}) &= \dfrac{1}{K - 1} \sum_{k=0}^{K-1}  (o_k^{(i)} - \hat o^{(i)})^2
\end{align}

$\hat o^{(i)}$ is the sample mean of the offspring of the $i^{\text{th}}$ particle in the $K$ offspring sequences. $\text{Var}(o^{(i)})$ is the sample variance in offspring of the $i^{\text{th}}$ particle in the $K$ offspring sequences. The MSE results are normalised by the number of particles, $\text{MSE}(\bm o)/N$. Then the contribution of the squared bias to the mean squared error given in~\eqref{eq:bias-contrib} is used to assess the bias of the resampling algorithms.

\begin{align}\label{eq:bias-contrib}
\text{Bias contribution}=||\text{Bias}(\bm o)||^2/\text{MSE}(\bm o) 
\end{align}

\vspace{2mm}
\noindent\textbf{Execution Time.~}To assess the speed of each resampling algorithm, we calculate the execution time of the kernels. We compute $B$ using ~\eqref{eq:bcalc} with $\epsilon = 0.01$ and ignore this computation time as it is the same for all of the algorithms, and, in practice, it can be chosen from some estimate of the weight distribution.

For each resampling algorithm, we consider the performance when the number of particles is increased from $2^6$ to $2^{22}$. We generate 16 unique weight sequences from both weight generation methods for each choice of the number of particles. Each algorithm under test performs 256 Monte Carlo runs ($K=256$) on each weight sequence, generating an offspring vector each time. The experimental results for a given weight distribution are the average across the 256 repeated assessments of the 16 generated weight sequences.

\begin{figure*}
    \centering
    \includegraphics[width=0.9\textwidth]{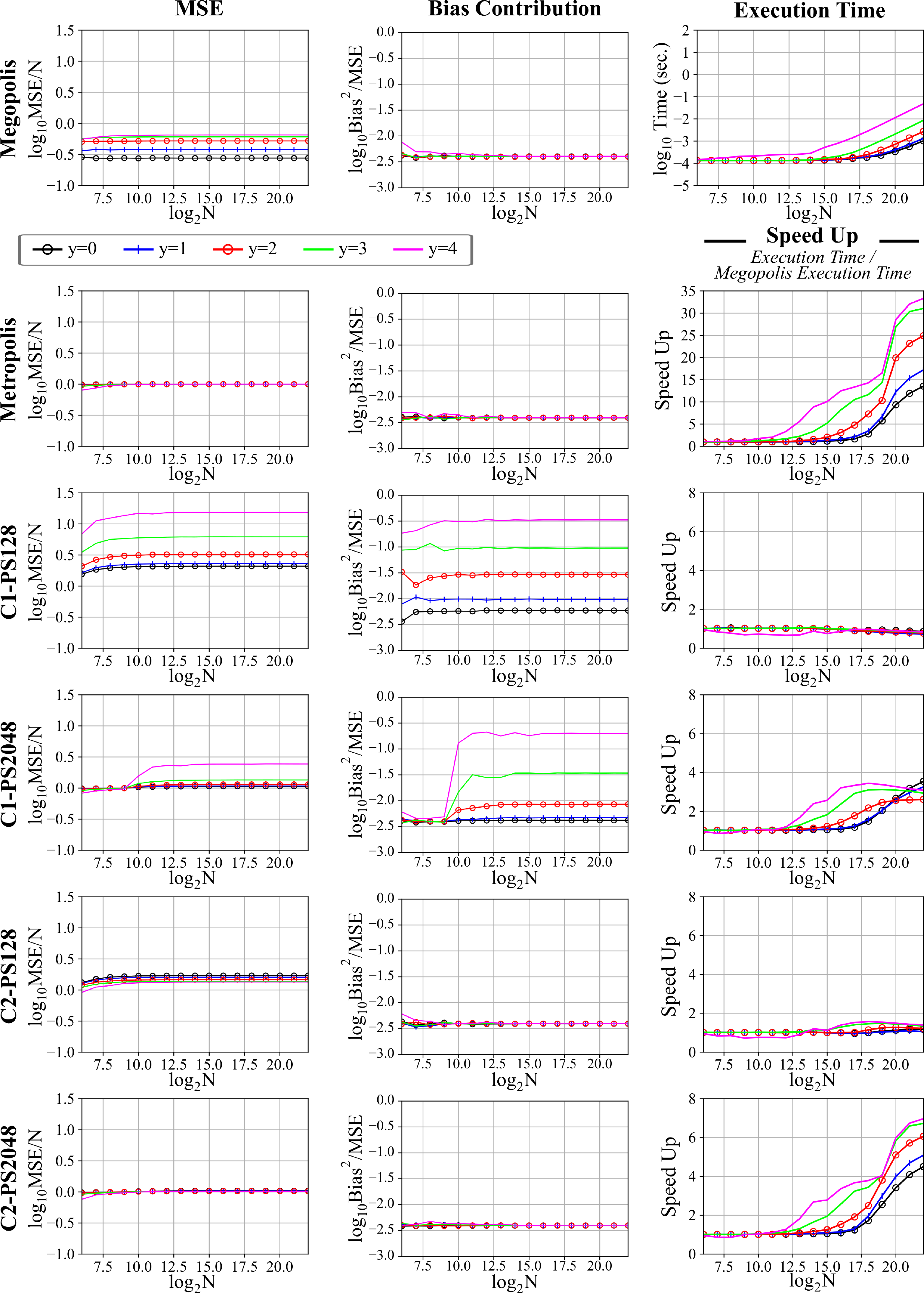}
    \caption{Comparison of experimental results of Megopolis, Metropolis, C1-PS128, C1-PS2048, C2-PS128, and C2-PS2048, using the distribution from  \eqref{eq:distribution} to generate weights. The speedup graphs directly compare the execution time of a given method with Megopolis where $\text{Speed~Up} = Execution Time / Megopolis~ Execution~Time$.}
    \label{fig:rvResults}
\end{figure*}

\section{Results and Discussion}\label{sec:resultsDiscussion}
This section compares the Megopolis algorithm with Metropolis and C1 and C2 with partition sizes 128 and 2048 bytes. For the sake of brevity, C1-PS128, C1-PS2048, C2-PS128, C2-PS2048 will be used to discuss C1 and C2 with the respective partition sizes. We compare the algorithms using the weight distribution described in~\eqref{eq:distribution} to generate weight sequences. The results are summarised in Fig.~\ref{fig:rvResults}. We also compare the algorithms with weight sequences generated by sampling from the gamma distribution described by~\eqref{eq:gammaEq}; however, as the results are very similar to those observed with the previous weight-generation method, we defer these results to \textbf{Appendix}~\ref{appendix:a}. 
To provide a baseline for the expected MSE and bias, we also compare Megopolis to the unbiased parallel resampling methods that require a prefix sum, the parallel multinomial algorithm, \cite{murray_gpu_2012}, and the improved parallel systematic method \cite{nicely2019improved}. These algorithms are provided in \textbf{Appendix}~\ref{appendix:unbiased} with a discussion in Section~\ref{sec:unbiasedDiscussion} and results summarised in Fig.~\ref{fig:unbiasedComparison}.

\subsection{MSE and Bias}\label{sec:mse-bias-results} 
The MSE and bias contribution results indicate that the Megopolis algorithm produces the highest quality results compared to Metropolis, C1, and C2. It produces a lower MSE than these algorithms without increasing the bias contribution to MSE. To understand the rationale for the lower MSE, we analyse the bias and variance of the Megopolis results separately. 

The bias contribution results in Fig.~\ref{fig:rvResults} demonstrate that the bias contribution for Megopolis is the same as both Metropolis and C2 but significantly lower than C1. We observe no increase in the bias for Megopolis due to the behaviour exhibited by an individual particle during ancestor selection. For an individual particle, at each iteration of the inner loop, the particle uses a random offset from $0$ to $N-1$ to select a particle for comparison with its current ancestor. As such, this choice can be any particle from the entire set of particles with equal probability, resulting in negligible bias values. Given that the MSE is composed of variance and bias, the lower MSE for Megopolis indicates that Megopolis produces less variance in the particle offspring across repeated resamples of the same distribution.

The lower variance in particle offspring for Megopolis is attributed to selecting a weight for comparison with the current ancestor. In Megopolis, each particle is exposed exactly once for weight comparison on $B$ occasions due to the global offsets used. As such, the maximum number of offspring a particle can have in a given resampling process is $B$; hence, the range of possible offspring a particle can have $0$ to $B$ for Megopolis. In contrast, in Metropolis, C1 and C2, each particle generates its own random index at each iteration; hence, for a single iteration, a single particle may be selected as an ancestor up to $N$ times while others may not be selected at all. Therefore, the range of possible offspring a particle can have is $0$ to $N$ for Metropolis, C1, and C2. As $B$ is generally significantly smaller than $N$, the range of possible offspring a particle can have is smaller in Megopolis than Metropolis, C1, and C2. Consequently, Megopolis produces significantly lower variance in particle offspring.

\subsection{Execution Time}
We evaluate and report the execution time of the proposed Megopolis and the relative speedup obtained with respect to other algorithms. The results demonstrate that Megopolis is significantly faster than Metropolis and C1 and C2 implementations with large partition sizes (2048). The execution time improvements become more prevalent when the number of particles exceeds $2^{14}$; at this point, the contribution of the kernel launch time to the total execution time becomes less significant.  Overall, as expected, Megopolis using the wrapped sequential memory access pattern to coalesce memory accesses attains significant performance improvements when the number of particles required is large.

Interestingly, small partition (PS128) implementations of C1 and C2 have similar execution time performance to Megopolis, with C1-PS128 marginally faster than Megopolis. The comparable execution times attained can be attributed to the similar number of memory transactions that Megopolis, C1-PS128, and C2-PS128 require. Megopolis requires 4 transactions at each iteration of the inner loop as it achieves the wrapped sequential memory access. The impact of the smaller partition sizes leads C1-PS128 and C2-PS128 to require at most 4 transactions. The marginal speedup of C1-PS128 over Megopolis can be attributed to C1-PS128 having improved cache utilisation as each thread warp uses the same partition for all iterations. Using the same partition at each iteration allows for better cache utilisation when each partition is small enough to fit entirely within the cache and, thus, reducing the need to access slower global memory. However, the speedup is marginal since C1 still requires generating two random numbers per particle (lines 8 and 9) for each iteration $B$ compared to the single random number generation needed in Megopolis (line 12). 

Despite requiring only at most 4 transactions in C2-PS128, Megopolis gains a small speedup over C2-PS128. This can be attributed to the behaviour C2 exhibits as the number of iterations $B$ increases. As $B$ increases, C2 must generate more random numbers than Megopolis for each iteration (three numbers per particle compared to the single number in Megopolis), resulting in longer execution times. We observe this effect in the speedup graph of C2-PS128 as $y$ increases. This is because increasing $y$ causes the number of iterations, $B$, to increase. The relationship between $B$ and $y$ is discussed in Section.~\ref{sec:effectOfWeightDist}.

Notably, despite the comparable execution times, the use of small partition sizes with C1 (C1-PS128) leads to significantly worse resampling quality results. This is indicated by C1-PS128 and C2-PS128 attaining a larger MSE compared to Megopolis for all $y$ weight distribution parameters (as we discussed in Section~\ref{sec:mse-bias-results}) and further confirmed through our results from the gamma distributions (in \textbf{Appendix}~\ref{appendix:a}).

In contrast, for larger partition sizes (e.g. PS2048), the number of memory transactions required by a warp for C1 and C2 increases to at most 32 (one per thread in a warp) because of the increased particle selection range. Megopolis achieves a significant speedup due to requiring only $4$ memory transactions per warp. 
\subsection{Effect of Weight Distribution Variations}
\label{sec:effectOfWeightDist}
Comparing the performance of the algorithms as the $y$ parameter increases (increasing coefficient of variation and altering the weight distribution) provides useful insight into the behaviour of each algorithm. Megopolis attains consistently low MSE and bias contribution results across all tested particle distributions; further confirmed through our results from the gamma distributions in \textbf{Appendix}~\ref{appendix:a}.

It is evident from Fig. \ref{fig:rvResults} that the MSE, bias contribution, and execution time of all the algorithms can be impacted by an increasing $y$ (an increasing concentration of weights on a few particles). Importantly, consider the effect an increasing $y$ has on the number of iterations $B$. From \eqref{eq:distribution}, we have $w^{(p)} = 1/\sqrt{2\pi}$ and  $\mathbb{E}(\boldsymbol{w}) = \exp{(-y^2/4)}/\sqrt{4\pi}$~\cite{murray2016parallel}. Thus, an increasing $y$ reduces the expected weight $\mathbb{E}(\boldsymbol{w})$. Additionally, since the maximum weight $w^{(p)}$ is a constant, based on \eqref{eq:bcalc}, we can see that $B$ is inversely proportional to $\mathbb{E}(\boldsymbol{w})$. Therefore, when $y$ increases, the weights become more concentrated on a few particles. 
Consequently, achieving a desirable error $\epsilon$ requires increasing the number of iterations $B$ as $y$ is increased in the sampled distribution. Although we have not directly evaluated the impact of $B$, we can see that the number of iterations $B$ affects the execution time of all algorithms. The speedup section of Fig. \ref{fig:rvResults} for an increasing $y$ (corresponding to an increasing $B$ value) illustrates that the impact from increasing iterations is the lowest on the proposed Megopolis and C1-PS128 algorithms. Thus we can expect Megopolis to be robust to changes in particle distributions and always generate higher quality results (seen with Megopolis generating the lowest MSE results for all $y$ distribution parameters) for a chosen number of iterations with significant or comparable speedup to Metropolis, C1 and C2 variants. 
Notably, C1 results demonstrate an increasing MSE and bias contribution as $y$ increases. This is due to the particles with large weights not being adequately exposed to random selections. In contrast, C2 shows a lower MSE than C1 and a constant bias contribution for small and large partition sizes. The decreasing MSE observed in C2 compared to C1 can be attributed to particles being exposed to more partitions as $B$ increases.

Interestingly, although the MSE of Megopolis is the lowest compared to all other algorithms, the MSE also increases with $y$.  However, the bias contribution remains constant. This suggests that the variance in the output sequence generated is increasing with $y$. As discussed in Section.~\ref{sec:mse-bias-results}, the range of offspring a particle can have in Megopolis is $0$ to $B$ while in Metropolis, the range is $0$ to $N$. As $B$ is increased to resample distributions with an increasing $y$, the particle offspring range, and consequently, the variance in the output sequence increases with $y$. In contrast, Metropolis maintains a constant, albeit higher, MSE. This can be attributed to the $0$ to $N$ particle offspring range unaffected by an increasing $B$. Importantly, the highest MSE achieved by Megopolis is still lower than the Metropolis algorithm under the shift in weights generated by an increasing $y$ parameter.

\subsection{Effect of C1 and C2 Partition Size Selection}

As reported~\cite{dulger_memory_2018}, we also observe that as the chosen C1 and C2 partition size approaches the total number of particles, the algorithms behave closer to Metropolis, improving both bias and MSE. However, increasing the partition size impacts the execution time of C1 and C2. An increased partition size requires an increased number of transactions to read a given partition from memory. More memory transactions result in increased execution time of the algorithms. Moreover, as the partition sizes are increased, the MSE results of C1 and C2 approach that of Metropolis. Therefore, we opted to compare the MSE and execution time of Megopolis against C1 and C2 with varying partition sizes in an effort to investigate the partition size at which both the C1 and C2 algorithms achieve an MSE closer to that of Megopolis. Subsequently, we compared the execution times of C1 and C2 at these \textit{optimal} partition sizes.

We evaluated C1 and C2 on the distribution from (\ref{eq:distribution}) with an exhaustive combination of the $y$ parameter chosen from $0, 1, 2, 3, 4$ and the number of particles in the range $2^{14}$ to $2^{22}$. We ignore the number of particles lower than $2^{14}$ as the choice of algorithm does not significantly impact execution time when the number of particles required to estimate the distribution is small. We evaluate the C1 and C2 algorithms using five different partition sizes,  $128, 256, 512, 1024, 2048$. The comprehensive set of results for all the experiments are in Appendix~\ref{appendix:b}. Fig. \ref{fig:BiasRuntime} presents a summary comparison of the MSE and execution time of Megopolis with respect to C1 and C2 with varying partition sizes when the number of particles is $2^{22}$, and the $y$ parameter is $4$. We choose a $y$ parameter of $4$ to demonstrate the resampling quality when the weight distribution is highly concentrated on a small number of particles, simulating the particle degeneracy problem.

\begin{figure*}
    \centering
    \includegraphics[scale=0.65]{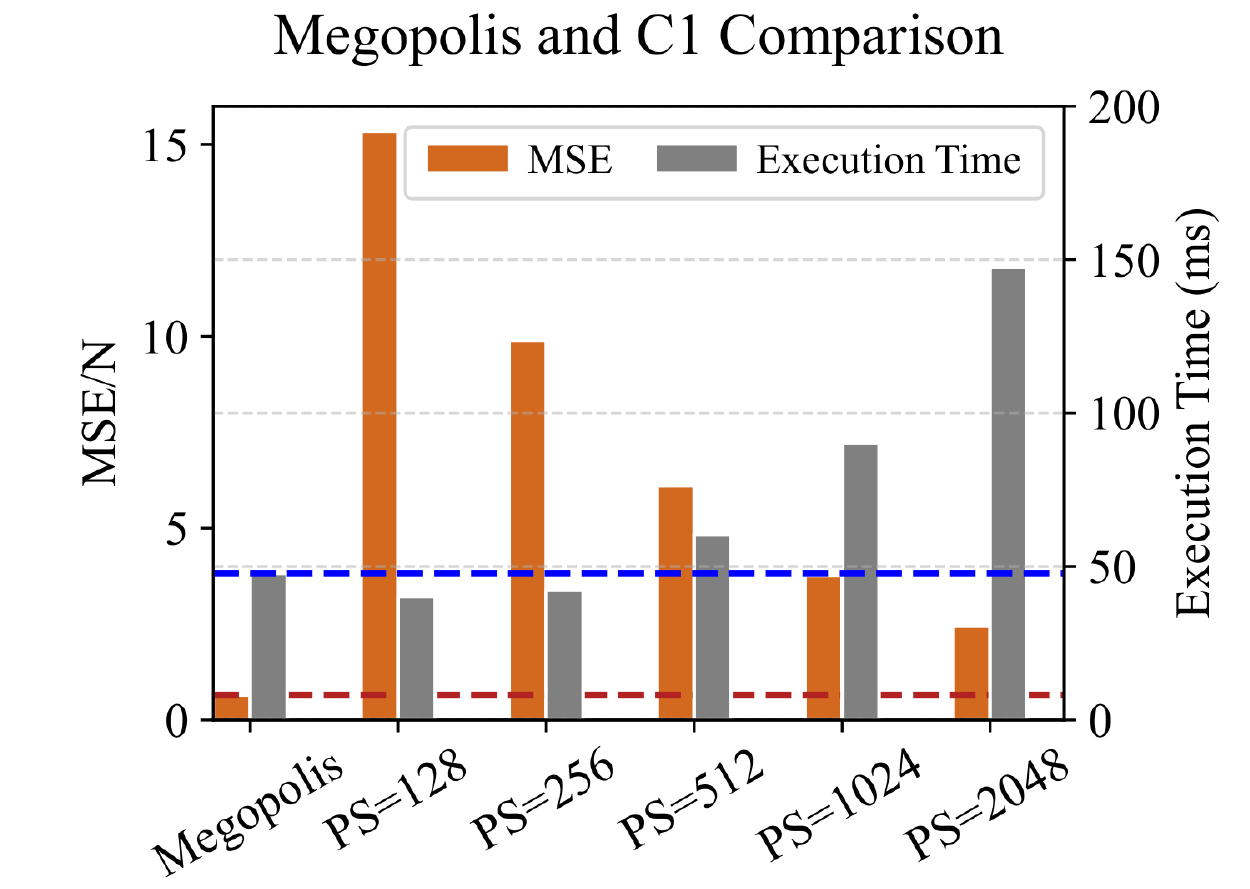}\hspace{0.4cm}
    \includegraphics[scale=0.65]{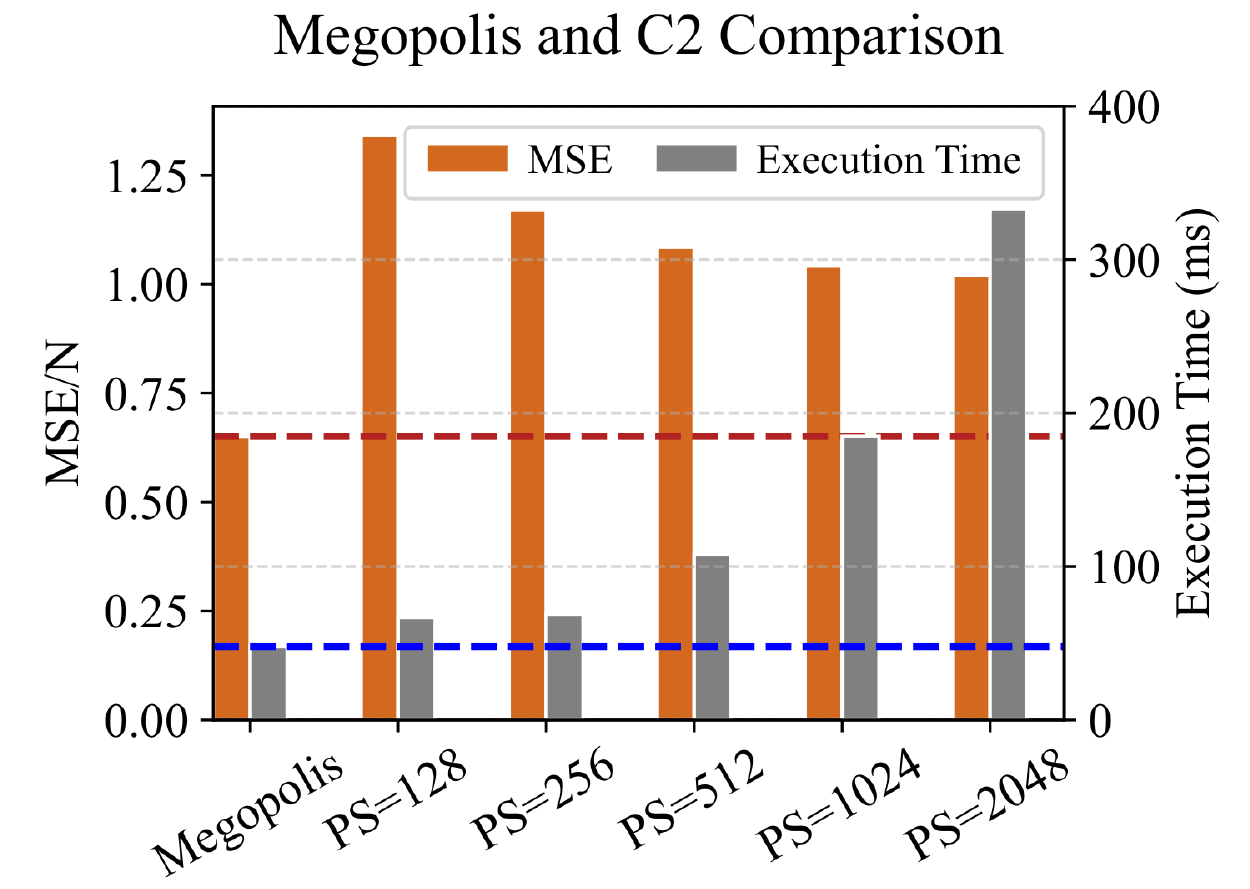}
    \caption{A comparison of the MSE and execution time of Megopolis with C1 and C2 with varying partition sizes when the number of particles is $2^{22}$ and the weights are sampled from the distribution from \eqref{eq:distribution} with $y = 4$. The dashed red and blue lines represent the MSE and execution time of Megopolis respectively.}
    \vspace{-0.5cm}
    \label{fig:BiasRuntime}
\end{figure*}

The results in Fig. \ref{fig:BiasRuntime} demonstrate that for any chosen partition size, resampling with Megopolis yields a lower MSE than C1 and C2. We also observe that Megopolis performs faster than C2 in all cases. The execution times for C1 is slightly lower for the smaller partition sizes of 128 or 256. However, with smaller partition sizes, we observe the artefacts of the C1 and C2 algorithms producing significantly higher MSE, resulting in lower quality resampling results. This experiment shows that C1 generates almost $15$ times increase in the MSE (see PS=128) compared to Megopolis. Consequently, we can observe that Megopolis will always produce a lower MSE and bias for a given execution time budget compared to C1 and C2; importantly, without needing to determine an appropriate partition size.

\subsection{Comparison with Prefix Sum Methods}\label{sec:unbiasedDiscussion}
\vspace{0.4cm}
We also compare Megopolis with the unbiased parallel resampling methods that require a prefix sum: i)~the parallel multinomial algorithm~\cite{murray_gpu_2012}; and ii)~the improved parallel systematic method~\cite{nicely2019improved} to provide a baseline for MSE and bias. From the results in  Fig.~\ref{fig:unbiasedComparison}, we see that Megopolis has a lower MSE than multinomial but a higher MSE than systematic. It is well-known that the systematic resampling method can reduce the variance of particle offspring over the multinomial method \cite{douc2005comparison}, hence, the MSE of the systematic method is lower than the multinomial method. Importantly, when comparing the bias contribution, the effect of numerical instabilities expected from the prefix sum can be observed as the bias contribution of both the multinomial and systematic methods increases with the number of particles. In contrast, the bias contribution of Megopolis is unaffected by the number of particles.

Comparing the execution time of Megopolis with multinomial and systematic in Fig.~\ref{fig:unbiasedComparison}, Megopolis achieves a significant speedup over both algorithms as the number of particles increases. Interestingly, we can observe the speedup Megopolis attains to be lower at higher $y$ values where the weights become highly concentrated on fewer particles. This observation is due to the number of iterations that Megopolis must perform, i.e., the complexity of Megopolis increasing with $y$ faster than the complexity of the multinomial and systematic algorithms for our choice of error $\epsilon = 0.01$--- see Equation~\eqref{eq:proposition}.

\begin{figure*}
    \centering
    \includegraphics[width=0.9\textwidth]{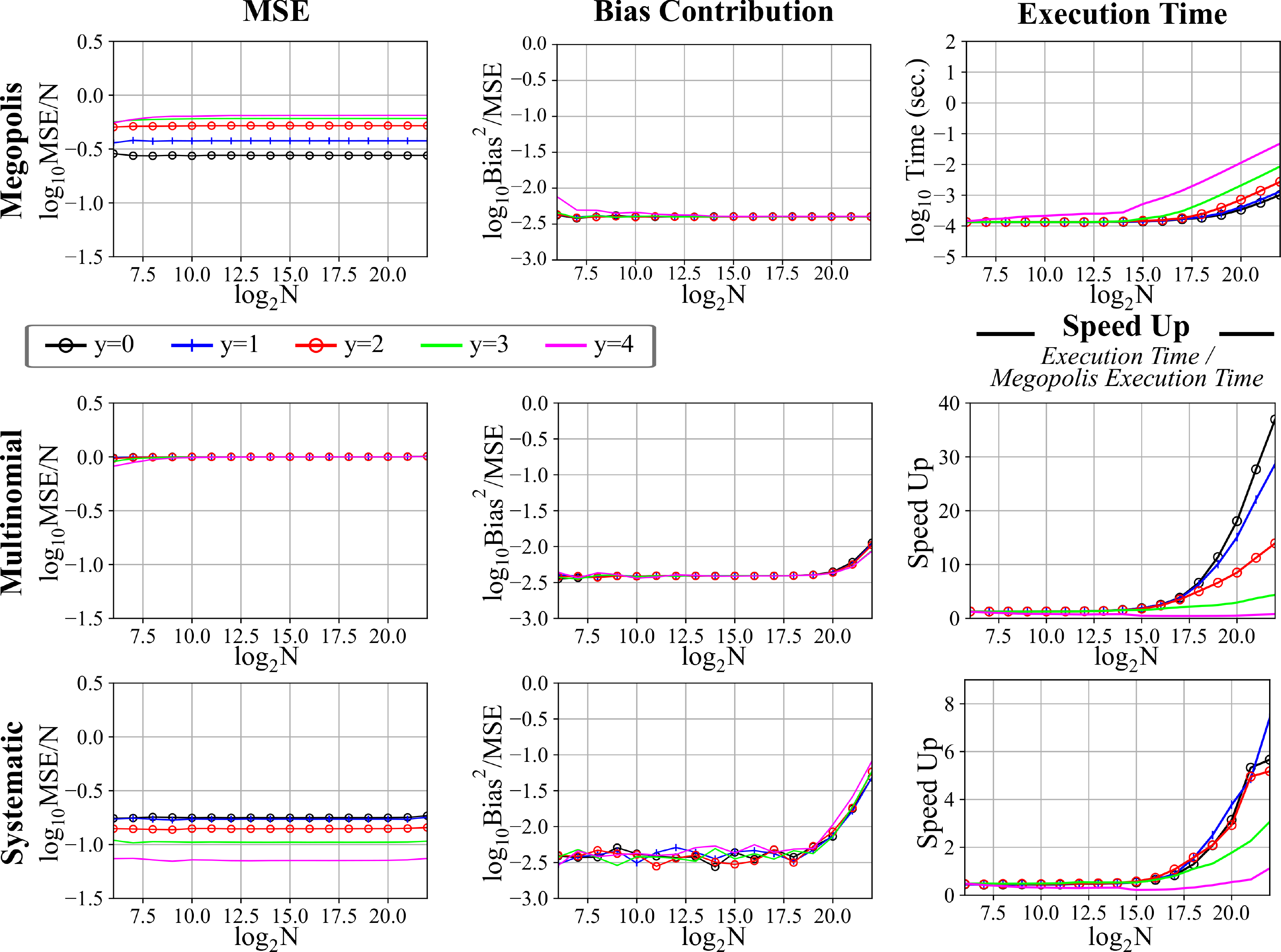}
    \caption{Comparison of experimental results of Megopolis, parallel Multinomial~\cite{murray_gpu_2012}, and  parallel Systematic~\cite{nicely2019improved} using the distribution from  \eqref{eq:distribution} to generate weights. The speedup graphs directly compare the execution time of a given method with Megopolis where $\text{Speed~Up} = Execution Time / Megopolis~ Execution~Time$.}
    \label{fig:unbiasedComparison}
\end{figure*}

\section{End-to-End Application Benchmark}~\label{sec:end_to_end_app}
We apply an SIR particle filter to a well known, highly non-linear  system~\cite{carlin1992a,gordon_novel_1993,kitagawa1996monte,arulampalam2002a} to benchmark the resampling algorithms in the context of a filtering application. For the SIR filter described in Section~\ref{sec:SIR}, the predict (propagation) and update equations for the non-linear system are as follows:
\begin{align}
    x_{t} &= \dfrac{x_{t-1}}{2} + 25 \dfrac{x_{t-1}}{1 + x^2_{t-1}} + 8 \cos{(1.2 t)} + v_{t-1}, \label{eq:predict-end2end}\\
    z_{t} &= \dfrac{x^2_t}{20} + n_t, \label{eq:update-end2end}
\end{align}
where $t$ is the time step; $x_t$ is the state; $z_t$ is the measurement; $v_{t-1}$ and $n_t$ are zero-mean Gaussian random variables with variance $o^2_v = 10$ and $o^2_n = 1$. For the SIR particle filter, we modify Algorithm \ref{alg:sirpf} to remove the weight normalisation step as the resampling algorithms used do not require normalised weights. As the output of the particles from resampling has uniform weighting, we can shift the estimation step to occur after resampling. In this case, we simply need to calculate the mean of the particles as our estimate. We use the efficient parallel reduction algorithm provided by the CUDA C++ \texttt{Thrust} library \cite{noauthor_cuda_nodate} to calculate the mean. The modified SIR particle filter is described in Algorithm \ref{alg:modsirpf} where $\hat{x}_{t}$ is the filter estimate at time step $t$, and the rest of the parameters are as described in Algorithm \ref{alg:sirpf}.  

\begin{algorithm}[!tb]
	\footnotesize
	\caption{Modified SIR Particle Filter}     \label{alg:modsirpf}
	\begin{algorithmic}[1] 
		\Statex \textbf{Input}: $[\boldsymbol{\bar{x}}_{t-1}, z_t]$
		\Statex \textbf{Output}: $[ \boldsymbol{\bar{x}}_t,\hat{x}_{t},]$
		\LineCommentNoIdent{\texttt{Stage 1:Prediction and Update}}
		\For {$i\gets 0$ to $N-1$} 
		    \State $x_t^{(i)} = f_{t-1}(\bar{x}^{(i)}_{t-1},v_{t-1})$ \Comment{prediction using \eqref{eq:prediction}}
		    \State $w^{(i)}_t = p(z_t | x_t^{(i)})$ \Comment{update using \eqref{eq:update}}
		\EndFor
		\LineCommentNoIdent{\texttt{Stage 2: Resample}}
		\State $\boldsymbol{\bar{x}}_t =$ RESAMPLE($[\boldsymbol{x}_t, \boldsymbol{w}_t]$)
		\LineCommentNoIdent{\texttt{Stage 3: Estimation}}
		\State $\hat{x}_{t} = \big(\sum_{i=0}^{N-1} \bar{x}_t^{(i)}\big)/N$.
	\end{algorithmic}
\end{algorithm}

\begin{figure*}
    \centering
    \begin{subfigure}[t]{\linewidth} 
        \centering
        \includegraphics[scale=1.0]{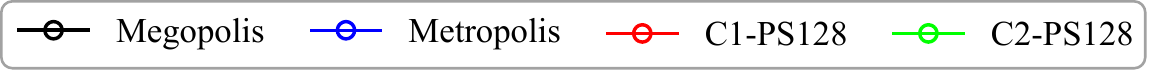}
    \end{subfigure}
    
    \begin{subfigure}[t]{0.28\linewidth}
        \includegraphics[scale=0.65]{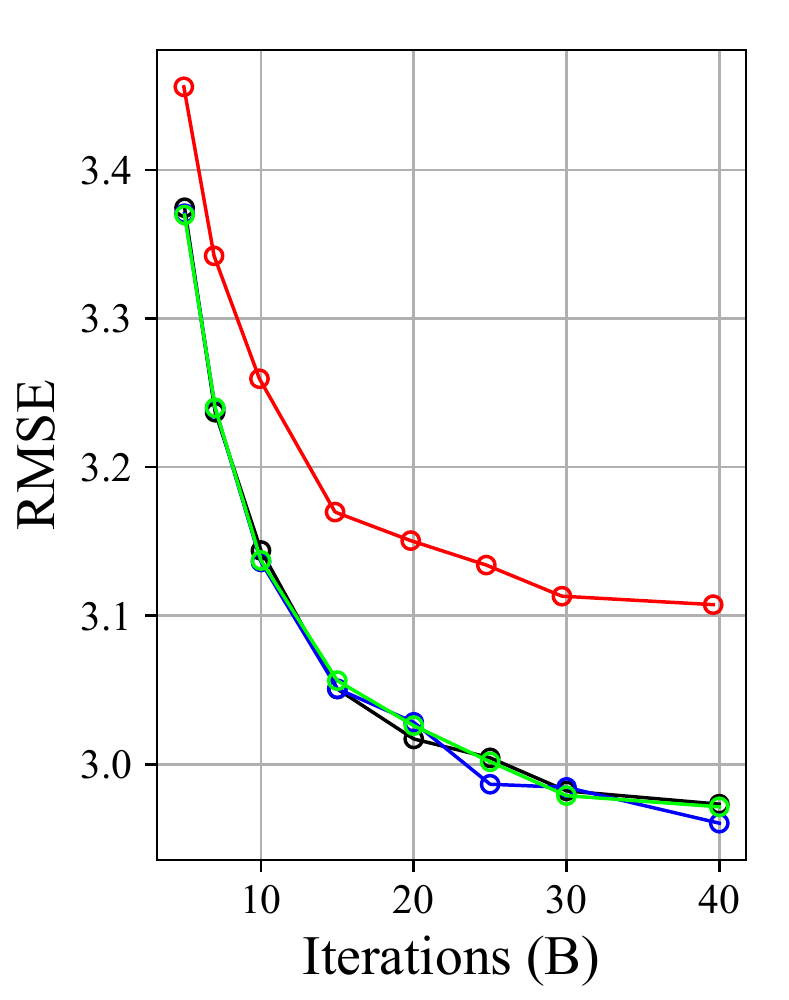}
        \caption{Mean RMSE results with varying $B$ parameter.}
        \label{fig:appRmse}
    \end{subfigure}
    \hspace{1cm}
    \begin{subfigure}[t]{0.28\linewidth}
        \includegraphics[scale=0.65]{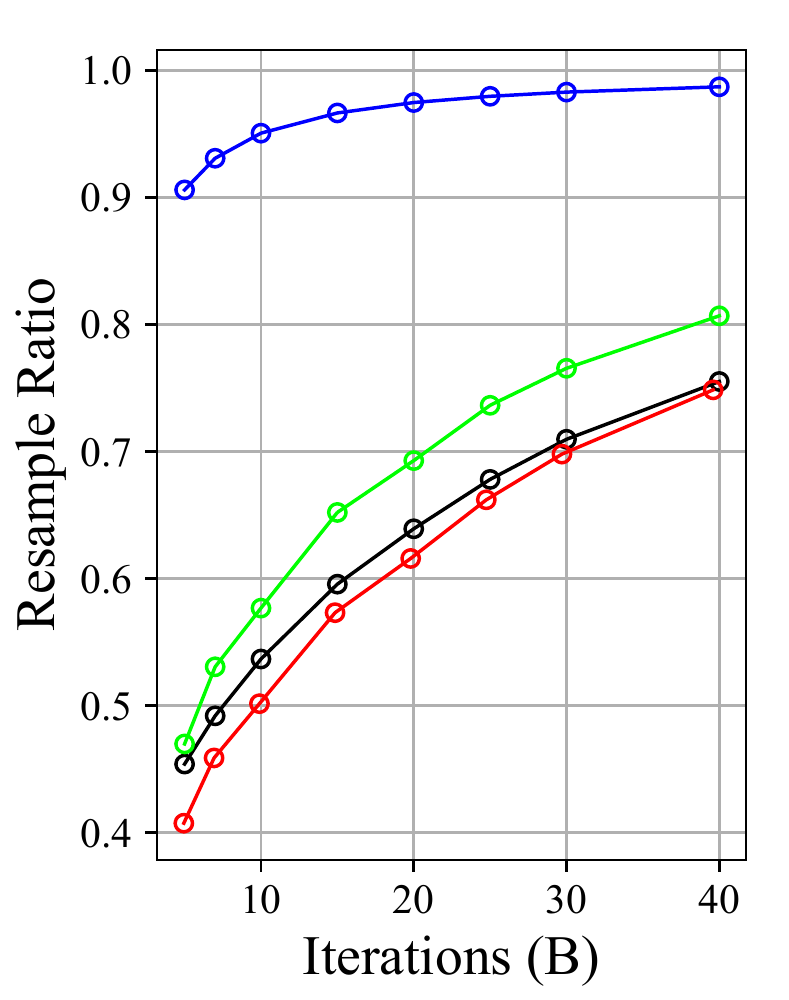}
        \caption{Resample ratio results with varying $B$ parameter.}
        \label{fig:appExecutionRatio}
    \end{subfigure}
    \hspace{1cm}
    \begin{subfigure}[t]{0.28\linewidth}
        \includegraphics[scale=0.65]{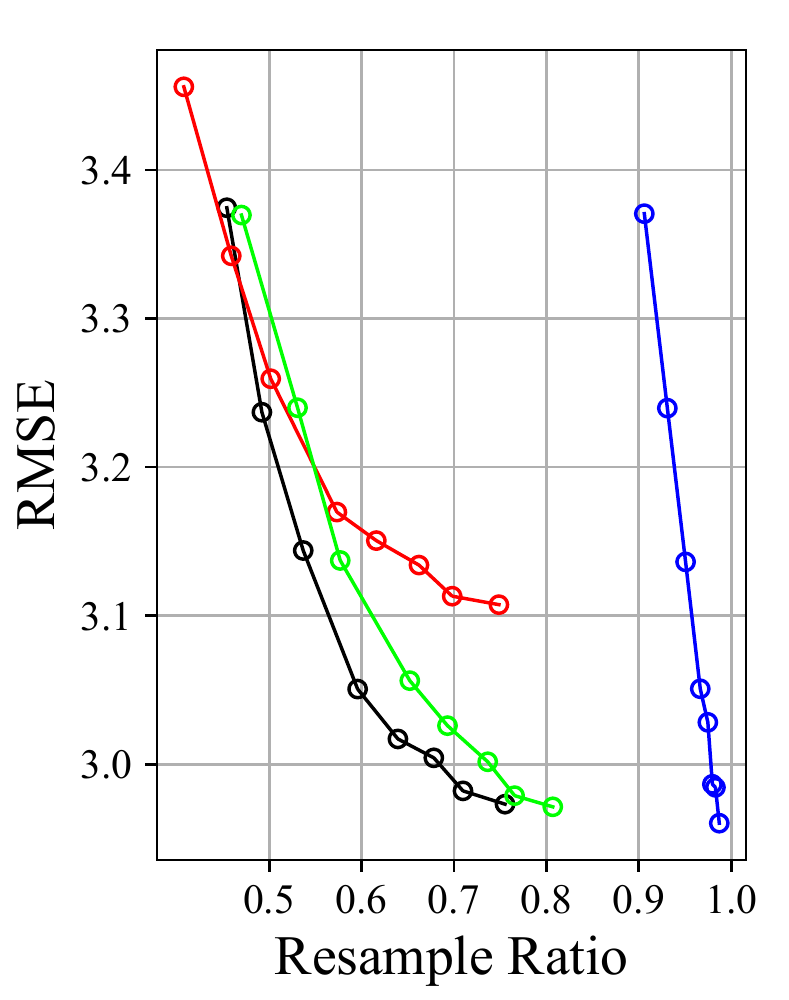}
        \caption{Mean RMSE results compared to the percentage of time spent resampling.}
        \label{fig:appBudget}
    \end{subfigure}

    \caption{End-to-end application results comparing mean RMSE, Resample Ratio, and number of iterations for Megopolis, Metropolis, C1-PS128, and C2-PS128.}
    \vspace{-0.5cm}
    \label{fig:appResults}
\end{figure*}

We generated 16 ground truth trajectories and, for each instance, performed 50 Monte Carlo experiments and executed the SIR filter over 100-time steps. At each time step, the error is calculated as the difference between the ground truth state and the estimated state. We use $2^{20}$ particles as the performance differences become more important with larger particle counts. To measure the quality of the estimation results generated from the resampling algorithms, we use the average root mean squared error (RMSE)~\cite{ristic_beyond_2003} of the generated trajectories, where RMSE is defined as follows: 
\begin{equation}
    \text{RMSE} = \dfrac{1}{T} \sum_{t=1}^T \sqrt{\dfrac{1}{K} \sum_{k=1}^K ||\hat{x}_t^{(k)} - x_t||^2},
\end{equation}
where $T=100$ is the number of time steps; $K=50$ is the number of Monte Carlo runs; $t$ is the time step; $x_t$ is the ground truth state at time step $t$; $\hat{x}_t^{(k)}$ is the estimated state at time step $t$ in the $k^{\text{th}}$ Monte Carlo run.

To assess time spent resampling, we compute the time spent in each stage of  Algorithm~\ref{alg:modsirpf}: i)~prediction and update time at stage $1$ ($\tau_{s_1}$); ii)~resample time at stage $2$ ($\tau_{s_2}$); and iii) estimation time at stage $3$ ($\tau_{s_3}$). We present the ratio of time spent resampling compared to the total execution time as:
\begin{equation}
    \text{Resample Ratio} = \dfrac{\tau_{s_2}}{\tau_{s_1} + \tau_{s_2} + \tau_{s_3}}.
    \label{eq:resampleRatio}
\end{equation}

We compare the results for Megopolis, Metropolis, C1-PS128, and C2-PS128 since  Megopolis demonstrates significant speedup over C1-PS2048 and C2-PS2048 while maintaining similar or better bias. As discussed earlier, choosing the number of resampling iterations $B$ for the algorithms significantly impacts execution time and resampling quality. In practical implementations, we want to avoid calculating $B$ at each resampling stage and employ a prior value chosen for the application context. Thus, we elect to evaluate the performance of the algorithms with varying $B$ parameters. First, we execute the application using each algorithm, calculating $B$ at runtime using \eqref{eq:bcalc} with an $\epsilon$ of $0.1$, from \cite{dulger_memory_2018}, \textit{before each resampling stage}. We choose a baseline for $B$ from the average of the calculated $B$ values. For this application, we calculated a baseline of $30$ iterations. Next, using this baseline, we evaluate the effects of varying $B$ by running the application for each algorithm with fixed $B$ values of $5, 7, 10, 15, 20, 25, 30, 40$. We focus on the effect of reducing $B$ as increasing $B$ provides diminishing improvements to RMSE, however, for completeness we include $B = 40$ to demonstrate this behaviour. We present the RMSE results in Fig.~\ref{fig:appRmse} and the resample ratio results in Fig.~\ref{fig:appExecutionRatio}. 

From the results in Fig.~\ref{fig:appRmse} and Fig.~\ref{fig:appExecutionRatio}, we can see that for any choice of $B$, C1-PS128 has the lowest resample ratio but a significantly higher RMSE value. Megopolis, Metropolis, and C2-PS128 have similar RMSE values, with Megopolis having the lowest resample ratio of the three. Importantly, we can see that the RMSE produced by C1-PS128 with a $B$ value of 25 is similar to the RMSE produced by Megopolis, Metropolis, and C2-PS128 with a $B$ value of 10.

Using the relationship between resample ratio and iterations $B$, we can construct a model for mapping resample ratio to RMSE for each algorithm. This is presented in Fig.~\ref{fig:appBudget}. Recall that the resample ratio is directly related to execution time; a lower resample ratio is equivalent to a lower execution time. Hence, this model shows the RMSE quality each algorithm can produce under \textit{varying execution time budget constraints}. From Fig.~\ref{fig:appBudget}, we can see that Megopolis produces the lowest RMSE out of all the algorithms for any given resample ratio. Consequently, under any chosen time constraint, these benchmark results demonstrate that the lowest RMSE value can be obtained by employing the Megopolis algorithm. 

In order to assess the impact of the bias exhibited by Megopolis, Metropolis, C1, and C2 on the RMSE, we compare these methods to the unbiased parallel multinomial algorithm~\cite{murray_gpu_2012} and improved parallel systematic~\cite{nicely2019improved}; both of which require a prefix sum. As the number of iterations $B$ affects the bias of the algorithms avoiding the prefix sum, we compare three choices of $B$, 16, 32, and 64 to provide a comparison with a range of biases.  We execute the application for each algorithm with $2^{20}$ particles. We selected $2^{20}$ particles to \textit{avoid} numerical instabilities produced by the prefix sum required for the unbiased methods. The resample ratio and RMSE results of the algorithms are presented in Table~\ref{tab:rmseTable}. The results confirm that our proposed Megopolis algorithm achieves a performance balance between the execution time (Resample Ratio) and accuracy (RMSE) without the need to evaluate and select a partition size, as in C1 and C2. Importantly, Megopolis achieves the best RMSE amongst the non-prefix sum methods; we can also observe that as the number of iterations $B$ is increased, Megopolis closely approaches the RMSE of unbiased resamplers (multinomial and systematic) whilst demonstrating a speedup over these unbiased methods.

\begin{table}[!hb]
\caption{End-to-end application results comparing resample ratio, and RMSE for the different resampling algorithms with $2^{20}$ particles. The best results out of Megopolis, Metropolis, C1-PS128, and C2-PS128 for each choice of $B$ are in \textbf{bold}, while the second-best results of those algorithms are in \secondbest{green}. }
\centering
\begin{tabular} {|c|l|c|c|c|c|}
\hline
$B$ & Type & Resample Ratio & RMSE \\
\hline
\hline
- & Multinomial \cite{murray_gpu_2012} & 0.925 & 2.944 \\
\hline
- & Systematic \cite{nicely2019improved} & 0.878 & 2.944 \\
\hline
\hline
\multirow{4}{*}{16} & Megopolis & \secondbest{0.603} & \textbf{3.039} \\
\cline{2-4}
& Metropolis & 0.968 & \secondbest{3.053} \\
& C1-PS128 & \textbf{0.573} & 3.170 \\ 
& C2-PS128 & 0.650 & \secondbest{3.053} \\
\hline
\hline
\multirow{4}{*}{32} & Megopolis & \secondbest{0.718} & \textbf{2.972} \\
\cline{2-4}
& Metropolis & 0.984 & \secondbest{2.978} \\
& C1-PS128 & \textbf{0.707} & 3.118  \\ 
& C2-PS128 & 0.776 & 2.985 \\
\hline
\hline
\multirow{4}{*}{64} & Megopolis & \secondbest{0.821} & \textbf{2.948} \\
\cline{2-4}
& Metropolis & 0.992 & \secondbest{2.956} \\
& C1-PS128 & \textbf{0.818} & 3.101  \\ 
& C2-PS128 & 0.867 & 2.957 \\
\hline
\end{tabular}
\label{tab:rmseTable}
\end{table}

\section{Conclusion and Discussion}~\label{sec:conclusion}
To improve coalesced memory access in the Metropolis resampling algorithm while also improving quality in output over other techniques, we developed the Megopolis algorithm. We proved that the number of iterations $B$ for our Megopolis algorithm is the same as for Metropolis to achieve the same error bound on the resampled distribution. The extensive experimental evidence demonstrates that the Megopolis algorithm always produces high-quality results with a lower MSE than the original Metropolis algorithm and both the C1 and C2 algorithms. The Megopolis algorithm is significantly faster than Metropolis as well as C1 and C2 when larger partition sizes are chosen to achieve lower MSE. Further, our extensive experimental results show that the execution time of the Megopolis algorithm is similar to the C1 and C2 algorithms, even when execution time is favoured over resampling quality by using smaller partition sizes. The smaller partition implementations for C1 achieve a marginal speedup over Megopolis at the cost of significantly increasing MSE and bias. However, through our end-to-end application benchmark, we demonstrate that by reducing the iterations for Megopolis to match the lower resampling quality of C1, Megopolis can achieve a speedup over C1. As for C2 with smaller partition implementations, Megopolis achieves a marginal speedup and improved MSE. Consequently, the Megopolis algorithm can be adapted to many application scenarios to provide both fast and high-quality resampling.

Although we have explored the application of our Megopolis resampler in the context of a SIR particle filter, in general, Megopolis can operate on weights (\eg, particle weights and mixture weights) that have not been normalised; an important feature for parallelisation shared with other Metropolis resampling algorithms. Notably, some sampling-based algorithms require resampling to be performed where the number of input samples is different from the number of output samples, \eg, the D-PMC~\cite{elvira2017population} method. To employ the Megopolis algorithm in this setting, the range of offsets in the Megopolis algorithm can be altered to fit the number of input samples. Further, other sampling-based algorithms require procedures that also pose challenging computational requirements, \eg, computing both the particle and mixture weights in the Improved Auxiliary Particle Filter (IAPF)~\cite{elvira2019elucidating}. We leave the evaluation of performance improvements the Megopolis algorithm can provide on a wider set of sampling-based algorithms and the exploration of methods to improve the performance of the other computationally burdensome procedures using massive parallelisation as directions to explore in future work.

Further, the cost of global memory transfers has been acknowledged by GPU manufacturers, and architectural changes have been made to lessen the impact. The performance of the Megopolis algorithm is greatly affected by global memory transfer speeds and would benefit from improvements to the size and rate of sequential memory transactions. On the other hand, improvements to random memory accesses would benefit Metropolis, C1, and C2 more than Megopolis. However, improvements to random memory accesses would likely only homogenise the execution time of the algorithms, and Megopolis can be expected to maintain an improved MSE over the other algorithms.

\section{Acknowledgement}
This work was supported by the Australian Research Council (LP160101177) 
and the Defence Science and Technology Group (DSTG), Australia.

\clearpage
\bibliographystyle{cas-model2-names}
\bibliography{references}
\clearpage
\bio{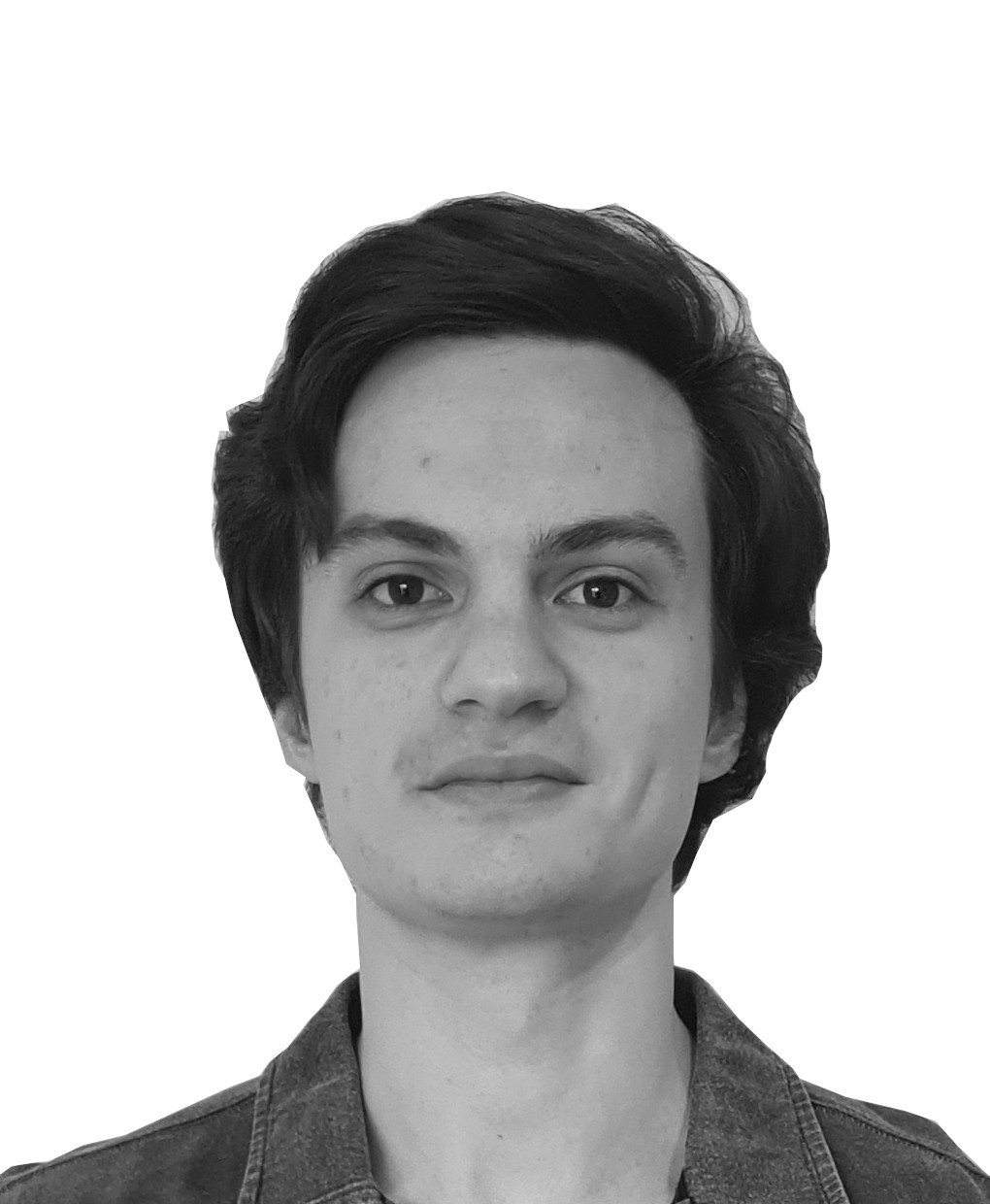}
\textbf{Joshua A. Chesser} received his B.Sc Advanced degree in Computer Science in 2020 from The University of Adelaide, Australia. He is currently pursuing his Honours degree in Computer Science and is a Research Associate with the School of Computer Science, The University of Adelaide. His research interests include robotics,  multi-object tracking and multi-sensor control. 
\endbio


\bio{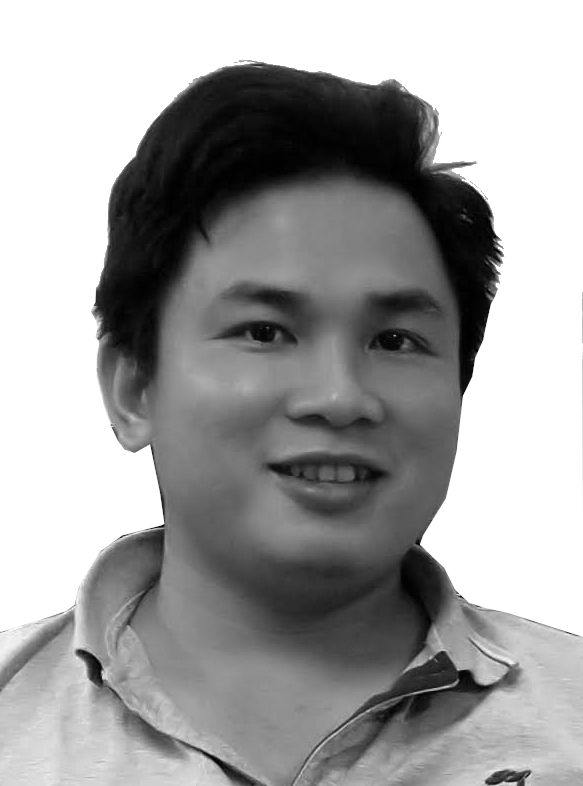}
\textbf{Hoa Van Nguyen} received his Bachelor degree in Electrical and Electronic Engineering from Portland State University, Oregon, the USA, in June 2012, and a PhD degree in Computer Science from The University of Adelaide, Australia, in July 2020. He is currently a Post-Doctoral Research Fellow with the School of Computer Science, The University of Adelaide. His research interests include signal processing, robotics, multi-object tracking and multi-sensor control.
\endbio


\bio{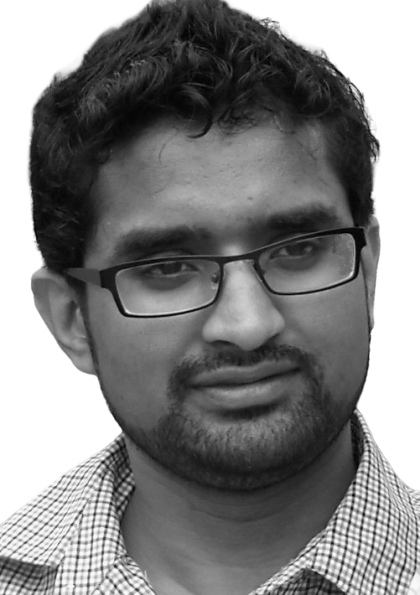}
\textbf{Damith C. Ranasinghe}  received the Ph.D. degree in Electrical and Electronic Engineering from The University of Adelaide, Australia, in December 2007. In the past, he was a Visiting Scholar with the Massachusetts Institute of Technology, Cambridge, MA, USA, and a Post-Doctoral Research Fellow with the University of Cambridge, Cambridge, U.K. Currently, he is an Associate Professor at The University of Adelaide. His research interests lie broadly in the areas of autonomous systems, machine learning, and systems security.
\endbio

\onecolumn
\clearpage
\appendix

\section{Appendix}
\label{appendix:a}
This appendix presents the comparison of experimental results of Megopolis, Metropolis, C1-PS128, C1-PS2048, C2-PS128, and C2-PS2048 when sampling from the gamma distribution to generate weight sequences.

\begin{figure}[b]
    \centering
    \includegraphics[width=0.8\textwidth]{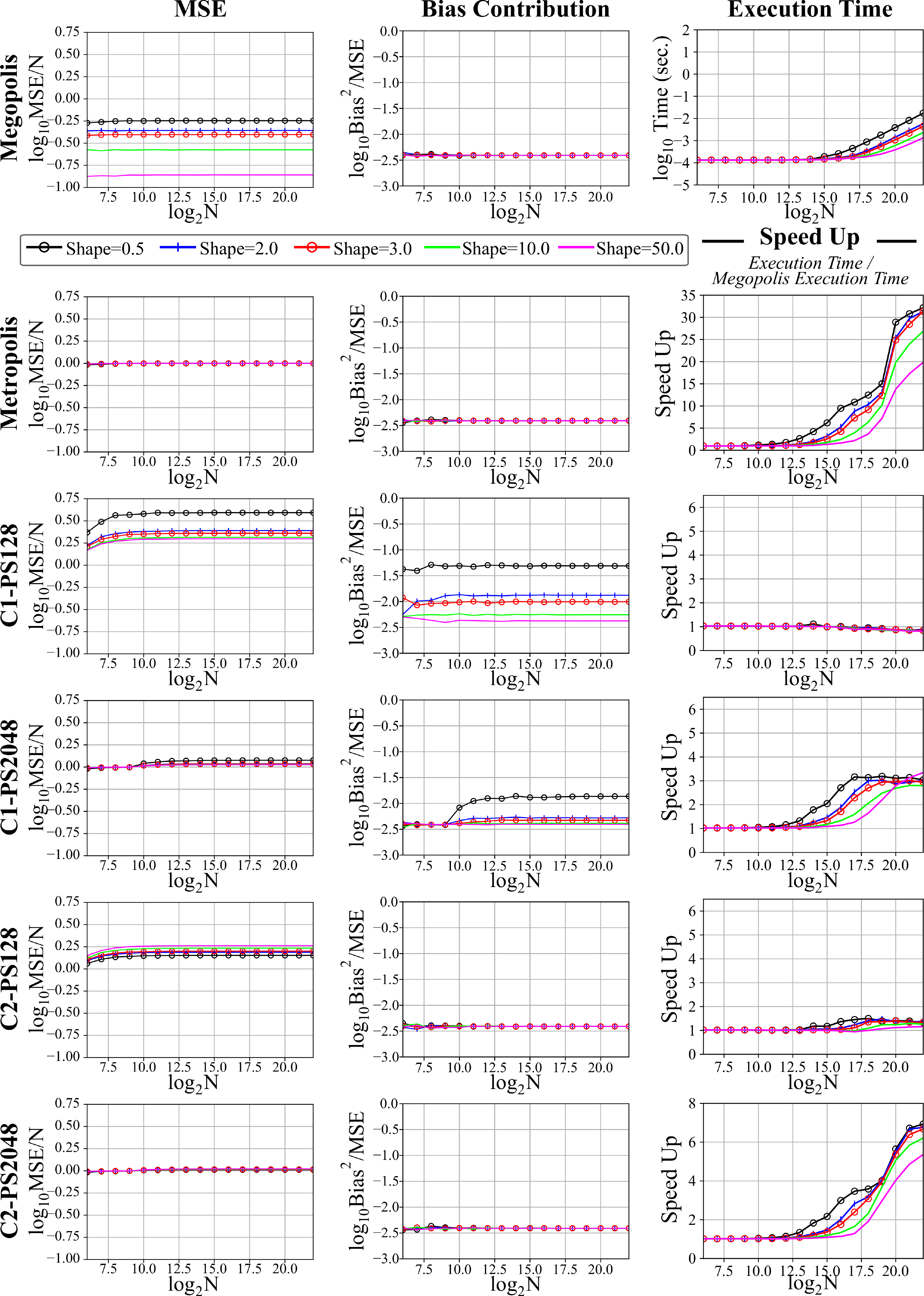}
    \caption{Comparison of experimental results of Megopolis, Metropolis, and C1 and C2 with partition sizes of 128 and 2048 bytes, using the gamma distribution to generate weights. The speedup graphs directly compare the execution time of a given method with Megopolis where $speedup = Execution~Time / Megopolis~ Execution~Time$. }
    \label{fig:gammaResults}
\end{figure}

\clearpage

\section{Appendix}
\label{appendix:unbiased}

\begin{algorithm}
	\footnotesize
	\caption{Multinomial Resample \cite{murray_gpu_2012}}     \label{alg:multinomialAlg}
	\begin{algorithmic}[1] 
		\Statex \textbf{Input}: $[\boldsymbol{x}_t, \boldsymbol{w}_t]$
		\Statex \textbf{Output}: $\bar{\boldsymbol{x}}_t$
		\State $\boldsymbol{\bar{w}}_t \gets \text{EXCLUSIVE\_PREFIX\_SUM}(\boldsymbol{w}_t)$
		
		\For {$i\gets 0$ to $N - 1$}
		    \State $u \sim \mathcal{U}[0,\boldsymbol{\bar{w}}_t^{(N-1)} + \boldsymbol{w}_t^{(N-1)}]$
		    \State binary search through $\boldsymbol{\bar{w}}_t$ to find $j$ such that $\boldsymbol{\bar{w}}_t^{(j)} \leq u < \boldsymbol{\bar{w}}_t^{(j+1)}$
		    \State $\bar{\boldsymbol{x}}_t^{(i)} \gets \boldsymbol{x}_t^{(j)}$
		\EndFor
		
	\end{algorithmic}
\end{algorithm}

\begin{algorithm}
	\footnotesize
	\caption{Improved Systematic Resample \cite{nicely2019improved}}     \label{alg:multinomialAlg}
	\begin{algorithmic}[1] 
		\Statex \textbf{Input}: $[\boldsymbol{x}_t, \boldsymbol{w}_t]$
		\Statex \textbf{Output}: $\bar{\boldsymbol{x}}_t$
		\State $\boldsymbol{\bar{w}}_t \gets \text{INCLUSIVE\_PREFIX\_SUM}(\boldsymbol{w}_t)$
		\State $u_t \gets \sim \mathcal{U}[0, 1]$
		
		\For {$i\gets 0$ to $N - 1$}
		    \State $u \gets (i + u_t) / N$
		    \State $m_t \gets \text{true}$ \Comment{thread bit mask}
		    \State $l_t \gets 0$
		    \State $a \gets i$ \Comment{ancestor index}
		    
		    \While{$m_t \neq \text{false}$}
		        \If{$i > (N - l_t)$}
		            \State $m_t \gets \text{false}$
		        \Else
		            \State $m_t \gets \boldsymbol{\bar{w}}_t^{(t + l_t)} < u$
                \EndIf		        
                \If{$m_t = \text{true}$}
                    \State $a \gets a + 1$
                \EndIf
                \State $l_t \gets l_t + 1$
		    \EndWhile \Comment{all thread bit masks must be false to exit}
		    
		    \State $l_t \gets 1$
		    
		    \While{$m_t \neq \text{true}$}
		        \If {$i < l_t$}
		            \State $m_t \gets \text{true}$
		        \Else
		            \State $m_t \gets \boldsymbol{\bar{w}}_t^{(t - l_t)} < u$
		        \EndIf
		        
		        \If {$m_t = \text{false}$}
		            $a \gets a - 1$
		        \EndIf
		        \State $l_t \gets l_t + 1$
		    \EndWhile \Comment{all thread bit masks must be true to exit}
		    \State $\bar{\boldsymbol{x}}_t^{(i)} \gets \boldsymbol{x}_t^{(a)}$
		\EndFor
		
	\end{algorithmic}
\end{algorithm}

\clearpage

\section{Appendix}
\label{appendix:b}

In this appendix, we summarise the bias and execution time results of each resampling algorithm on the distribution from (\ref{eq:distribution}), including Metropolis, Megopolis, and both Metropolis-C1 and Metropolis-C2 with partition sizes 128, 256, 512, 1024, and 2048 bytes in Tables 
\ref{tab:megopolisrv},
\ref{tab:metropolisrv},  \ref{tab:metropolisC1rv}, \ref{tab:metropolisC2rv}, respectively.

\begin{table*}
 \caption{The MSE and execution time results of the Megopolis algorithm on the distribution from \eqref{eq:distribution}. The results are displayed in the form "MSE/N, execution time" for each choice of $y$ and number of particles}
    \centering

\begin{tabular} {|c|c|c|c|c|c|}
\hline
\multicolumn{6}{|c|}{\textbf{Megopolis}} \\
\hline
\multirow{2}{*}{particles} & \multicolumn{5}{|c|}{$y$} \\
\cline{2-6}
& 0 & 1 & 2 & 3 & 4 \\
\hline
32768 & 0.2760, 1.396e-4 & 0.3769, 1.389e-4 & 0.5213, 1.475e-4 & 0.6067, 1.719e-4 & 0.6508, 5.170e-4 \\
\hline
65536 & 0.2753, 1.467e-4 & 0.3767, 1.496e-4 & 0.5214, 1.558e-4 & 0.6071, 2.086e-4 & 0.6507, 8.173e-4 \\
\hline
131072 & 0.2755, 1.652e-4 & 0.3767, 1.698e-4 & 0.5213, 1.798e-4 & 0.6072, 3.028e-4 & 0.6505, 0.0014 \\
\hline
262144 & 0.2756, 1.851e-4 & 0.3767, 1.943e-4 & 0.5214, 2.427e-4 & 0.6072, 5.410e-4 & 0.6510, 0.0027 \\
\hline
524288 & 0.2757, 2.292e-4 & 0.3768, 2.589e-4 & 0.5213, 3.928e-4 & 0.6070, 0.0010 & 0.6511, 0.0055 \\
\hline
1048576 & 0.2758, 3.377e-4 & 0.3768, 4.044e-4 & 0.5213, 7.156e-4 & 0.6071, 0.0021 & 0.6511, 0.0113 \\
\hline
2097152 & 0.2757, 5.643e-4 & 0.3768, 7.069e-4 & 0.5213, 0.0014 & 0.6071, 0.0042 & 0.6511, 0.0233 \\
\hline
4194304 & 0.2757, 0.0010 & 0.3767, 0.0013 & 0.5213, 0.0027 & 0.6071, 0.0086 & 0.6510, 0.0478 \\
\hline
\end{tabular}
   
    \label{tab:megopolisrv}
\end{table*}

\begin{table*}
    \centering
     \caption{The MSE and execution time results of the Metropolis algorithm on the distribution from \eqref{eq:distribution}. Each cell contains "MSE/N, execution time" for each choice of $y$ and number of particles}

\begin{tabular} {|c|c|c|c|c|c|}
\hline
\multicolumn{6}{|c|}{\textbf{Metropolis}} \\
\hline
\multirow{2}{*}{particles} & \multicolumn{5}{|c|}{$y$} \\
\cline{2-6}
& 0 & 1 & 2 & 3 & 4 \\
\hline
32768 & 0.9997, 1.620e-4 & 1.0001, 1.846e-4 & 1.0001, 3.013e-4 & 1.0001, 8.962e-4 & 0.9994, 0.0052 \\
\hline
65536 & 0.9998, 1.944e-4 & 1.0000, 2.406e-4 & 1.0000, 4.805e-4 & 1.0001, 0.0017 & 1.0000, 0.0102 \\
\hline
131072 & 0.9999, 2.719e-4 & 0.9999, 3.598e-4 & 1.0001, 8.597e-4 & 1.0000, 0.0032 & 0.9995, 0.0188 \\
\hline
262144 & 0.9999, 5.170e-4 & 1.0000, 6.696e-4 & 1.0000, 0.0018 & 1.0000, 0.0063 & 1.0002, 0.0386 \\
\hline
524288 & 0.9998, 0.0013 & 1.0000, 0.0017 & 1.0001, 0.0041 & 0.9999, 0.0150 & 1.0002, 0.0902 \\
\hline
1048576 & 0.9998, 0.0032 & 1.0000, 0.0050 & 1.0000, 0.0143 & 1.0001, 0.0556 & 1.0002, 0.3223 \\
\hline
2097152 & 0.9999, 0.0067 & 1.0000, 0.0109 & 1.0000, 0.0321 & 1.0001, 0.1264 & 1.0000, 0.7466 \\
\hline
4194304 & 0.9999, 0.0139 & 1.0000, 0.0227 & 1.0001, 0.0675 & 1.0001, 0.2672 & 1.0002, 1.5929 \\
\hline
\end{tabular}
    \label{tab:metropolisrv}
\end{table*}

\begin{table*}[!hb]
    \caption{The MSE and execution time results of the C1 algorithm on the distribution from \eqref{eq:distribution}. The results are displayed in the form "MSE/N, execution time" for each choice of $y$ and number of particles}
    \centering
\begin{tabular} {|c|c|c|c|c|c|}
\hline
\multicolumn{6}{|c|}{\textbf{Metropolis C1}} \\
\hline
\multirow{2}{*}{particles} & \multicolumn{5}{|c|}{$y$} \\
\cline{2-6}
& 0 & 1 & 2 & 3 & 4 \\
\hline
\multicolumn{6}{|c|}{Partition Size = 128} \\
\hline
32768 & 2.0902, 1.402e-4 & 2.3105, 1.407e-4 & 3.2380, 1.476e-4 & 6.2008, 1.738e-4 & 15.3599, 3.927e-4 \\
\hline
65536 & 2.0908, 1.429e-4 & 2.3090, 1.429e-4 & 3.2362, 1.521e-4 & 6.2169, 2.096e-4 & 15.4083, 7.143e-4 \\
\hline
131072 & 2.0906, 1.492e-4 & 2.3102, 1.496e-4 & 3.2421, 1.646e-4 & 6.2085, 2.915e-4 & 15.4024, 0.0014 \\
\hline
262144 & 2.0907, 1.713e-4 & 2.3113, 1.710e-4 & 3.2412, 2.126e-4 & 6.2096, 5.194e-4 & 15.4186, 0.0026 \\
\hline
524288 & 2.0910, 2.146e-4 & 2.3113, 2.116e-4 & 3.2410, 3.390e-4 & 6.2012, 9.563e-4 & 15.3472, 0.0051 \\
\hline
1048576 & 2.0908, 3.068e-4 & 2.3112, 3.153e-4 & 3.2410, 5.957e-4 & 6.2035, 0.0018 & 15.3344, 0.0101 \\
\hline
2097152 & 2.0910, 4.984e-4 & 2.3110, 5.321e-4 & 3.2413, 0.0011 & 6.2057, 0.0036 & 15.3497, 0.0202 \\
\hline
4194304 & 2.0910, 8.811e-4 & 2.3110, 9.659e-4 & 3.2413, 0.0021 & 6.2040, 0.0071 & 15.3533, 0.0402 \\
\hline
\multicolumn{6}{|c|}{Partition Size = 256} \\
\hline
32768 & 1.5457, 1.401e-4 & 1.6548, 1.406e-4 & 2.1225, 1.492e-4 & 3.7169, 1.795e-4 & 9.8313, 4.159e-4 \\
\hline
65536 & 1.5462, 1.451e-4 & 1.6564, 1.453e-4 & 2.1245, 1.568e-4 & 3.7268, 2.242e-4 & 9.9361, 7.820e-4 \\
\hline
131072 & 1.5458, 1.557e-4 & 1.6560, 1.565e-4 & 2.1222, 1.740e-4 & 3.7305, 3.253e-4 & 9.8162, 0.0014 \\
\hline
262144 & 1.5463, 1.793e-4 & 1.6561, 1.790e-4 & 2.1241, 2.308e-4 & 3.7353, 5.600e-4 & 9.9280, 0.0027 \\
\hline
524288 & 1.5463, 2.298e-4 & 1.6564, 2.318e-4 & 2.1233, 3.705e-4 & 3.7251, 0.0010 & 9.9072, 0.0054 \\
\hline
1048576 & 1.5463, 3.358e-4 & 1.6563, 3.481e-4 & 2.1237, 6.402e-4 & 3.7254, 0.0019 & 9.8904, 0.0106 \\
\hline
2097152 & 1.5465, 5.511e-4 & 1.6562, 5.849e-4 & 2.1238, 0.0012 & 3.7277, 0.0038 & 9.8972, 0.0212 \\
\hline
4194304 & 1.5464, 9.704e-4 & 1.6564, 0.0011 & 2.1235, 0.0022 & 3.7273, 0.0074 & 9.8994, 0.0423 \\
\hline
\multicolumn{6}{|c|}{Partition Size = 512} \\
\hline
32768 & 1.2729, 1.409e-4 & 1.3272, 1.421e-4 & 1.5576, 1.538e-4 & 2.3893, 2.035e-4 & 6.1284, 5.749e-4 \\
\hline
65536 & 1.2729, 1.480e-4 & 1.3278, 1.497e-4 & 1.5621, 1.673e-4 & 2.3934, 2.707e-4 & 6.0562, 0.0011 \\
\hline
131072 & 1.2732, 1.661e-4 & 1.3283, 1.684e-4 & 1.5620, 1.963e-4 & 2.3966, 4.177e-4 & 6.0780, 0.0020 \\
\hline
262144 & 1.2732, 1.958e-4 & 1.3280, 1.989e-4 & 1.5620, 2.792e-4 & 2.3944, 7.460e-4 & 6.1108, 0.0039 \\
\hline
524288 & 1.2735, 2.651e-4 & 1.3281, 2.714e-4 & 1.5621, 4.553e-4 & 2.3908, 0.0014 & 6.1052, 0.0077 \\
\hline
1048576 & 1.2732, 4.088e-4 & 1.3281, 4.227e-4 & 1.5626, 8.099e-4 & 2.3917, 0.0027 & 6.1060, 0.0152 \\
\hline
2097152 & 1.2733, 7.043e-4 & 1.3284, 7.252e-4 & 1.5623, 0.0015 & 2.3909, 0.0052 & 6.0944, 0.0302 \\
\hline
4194304 & 1.2733, 0.0013 & 1.3283, 0.0013 & 1.5622, 0.0029 & 2.3916, 0.0103 & 6.0995, 0.0602 \\
\hline
\multicolumn{6}{|c|}{Partition Size = 1024} \\
\hline
32768 & 1.1358, 1.426e-4 & 1.1636, 1.453e-4 & 1.2795, 1.625e-4 & 1.7024, 2.428e-4 & 3.7104, 8.112e-4 \\
\hline
65536 & 1.1361, 1.527e-4 & 1.1636, 1.562e-4 & 1.2802, 1.874e-4 & 1.6978, 3.485e-4 & 3.7609, 0.0016 \\
\hline
131072 & 1.1363, 1.781e-4 & 1.1637, 1.857e-4 & 1.2810, 2.408e-4 & 1.7000, 5.817e-4 & 3.7896, 0.0029 \\
\hline
262144 & 1.1366, 2.257e-4 & 1.1641, 2.386e-4 & 1.2815, 3.712e-4 & 1.7005, 0.0011 & 3.7432, 0.0057 \\
\hline
524288 & 1.1365, 3.319e-4 & 1.1641, 3.597e-4 & 1.2811, 6.316e-4 & 1.7018, 0.0020 & 3.7619, 0.0113 \\
\hline
1048576 & 1.1367, 5.618e-4 & 1.1642, 6.137e-4 & 1.2812, 0.0012 & 1.7012, 0.0039 & 3.7638, 0.0226 \\
\hline
2097152 & 1.1367, 0.0010 & 1.1642, 0.0011 & 1.2812, 0.0022 & 1.7021, 0.0078 & 3.7634, 0.0452 \\
\hline
4194304 & 1.1366, 0.0020 & 1.1642, 0.0022 & 1.2813, 0.0043 & 1.7018, 0.0155 & 3.7636, 0.0903 \\
\hline
\multicolumn{6}{|c|}{Partition Size = 2048} \\
\hline
32768 & 1.0675, 1.467e-4 & 1.0804, 1.523e-4 & 1.1380, 1.812e-4 & 1.3458, 3.162e-4 & 2.4094, 0.0013 \\
\hline
65536 & 1.0678, 1.594e-4 & 1.0811, 1.681e-4 & 1.1400, 2.234e-4 & 1.3525, 4.942e-4 & 2.4772, 0.0026 \\
\hline
131072 & 1.0680, 1.957e-4 & 1.0819, 2.112e-4 & 1.1403, 3.185e-4 & 1.3505, 8.859e-4 & 2.4354, 0.0047 \\
\hline
262144 & 1.0682, 2.760e-4 & 1.0820, 3.077e-4 & 1.1406, 5.289e-4 & 1.3523, 0.0017 & 2.4431, 0.0093 \\
\hline
524288 & 1.0684, 4.674e-4 & 1.0819, 5.305e-4 & 1.1406, 9.671e-4 & 1.3524, 0.0032 & 2.4417, 0.0184 \\
\hline
1048576 & 1.0682, 9.037e-4 & 1.0819, 0.0010 & 1.1405, 0.0018 & 1.3530, 0.0064 & 2.4472, 0.0368 \\
\hline
2097152 & 1.0682, 0.0018 & 1.0821, 0.0021 & 1.1406, 0.0036 & 1.3521, 0.0127 & 2.4476, 0.0737 \\
\hline
4194304 & 1.0682, 0.0036 & 1.0821, 0.0043 & 1.1406, 0.0071 & 1.3522, 0.0253 & 2.4407, 0.1475 \\
\hline
\end{tabular}
    
    \label{tab:metropolisC1rv}
\end{table*}

\begin{table*}[!hb]
\caption{The MSE and execution time results of the C2 algorithm on the distribution from \eqref{eq:distribution}. The results are displayed in the form "MSE/N, execution time" for each choice of $y$ and number of particles}
    \centering
\begin{tabular} {|c|c|c|c|c|c|}
\hline
\multicolumn{6}{|c|}{\textbf{Metropolis C2}} \\
\hline
\multirow{2}{*}{particles} & \multicolumn{5}{|c|}{$y$} \\
\cline{2-6}
& 0 & 1 & 2 & 3 & 4 \\
\hline
\multicolumn{6}{|c|}{Partition Size = 128} \\
\hline
32768 & 1.7029, 1.398e-4 & 1.6066, 1.398e-4 & 1.4663, 1.475e-4 & 1.3841, 1.967e-4 & 1.3396, 5.845e-4 \\
\hline
65536 & 1.7027, 1.461e-4 & 1.6061, 1.465e-4 & 1.4668, 1.568e-4 & 1.3839, 2.696e-4 & 1.3410, 0.0011 \\
\hline
131072 & 1.7031, 1.583e-4 & 1.6065, 1.604e-4 & 1.4673, 1.825e-4 & 1.3839, 4.263e-4 & 1.3416, 0.0022 \\
\hline
262144 & 1.7033, 1.852e-4 & 1.6063, 1.927e-4 & 1.4672, 2.782e-4 & 1.3841, 7.903e-4 & 1.3420, 0.0042 \\
\hline
524288 & 1.7034, 2.458e-4 & 1.6066, 2.675e-4 & 1.4673, 4.933e-4 & 1.3841, 0.0016 & 1.3413, 0.0084 \\
\hline
1048576 & 1.7034, 3.820e-4 & 1.6066, 4.325e-4 & 1.4670, 9.116e-4 & 1.3842, 0.0030 & 1.3413, 0.0168 \\
\hline
2097152 & 1.7035, 6.622e-4 & 1.6065, 7.724e-4 & 1.4672, 0.0017 & 1.3842, 0.0058 & 1.3414, 0.0334 \\
\hline
4194304 & 1.7034, 0.0012 & 1.6067, 0.0014 & 1.4673, 0.0033 & 1.3842, 0.0116 & 1.3414, 0.0668 \\
\hline
\multicolumn{6}{|c|}{Partition Size = 256} \\
\hline
32768 & 1.3508, 1.399e-4 & 1.3017, 1.406e-4 & 1.2322, 1.498e-4 & 1.1908, 2.019e-4 & 1.1696, 6.062e-4 \\
\hline
65536 & 1.3513, 1.473e-4 & 1.3021, 1.486e-4 & 1.2325, 1.615e-4 & 1.1906, 2.759e-4 & 1.1700, 0.0012 \\
\hline
131072 & 1.3514, 1.643e-4 & 1.3021, 1.695e-4 & 1.2326, 1.930e-4 & 1.1911, 4.430e-4 & 1.1700, 0.0022 \\
\hline
262144 & 1.3509, 1.998e-4 & 1.3026, 2.149e-4 & 1.2328, 2.997e-4 & 1.1913, 8.083e-4 & 1.1698, 0.0043 \\
\hline
524288 & 1.3512, 2.780e-4 & 1.3026, 3.168e-4 & 1.2327, 5.355e-4 & 1.1913, 0.0016 & 1.1701, 0.0086 \\
\hline
1048576 & 1.3511, 4.479e-4 & 1.3025, 5.361e-4 & 1.2329, 0.0010 & 1.1913, 0.0031 & 1.1699, 0.0172 \\
\hline
2097152 & 1.3512, 7.931e-4 & 1.3027, 9.732e-4 & 1.2328, 0.0020 & 1.1913, 0.0061 & 1.1699, 0.0344 \\
\hline
4194304 & 1.3511, 0.0015 & 1.3025, 0.0019 & 1.2328, 0.0039 & 1.1913, 0.0123 & 1.1700, 0.0689 \\
\hline
\multicolumn{6}{|c|}{Partition Size = 512} \\
\hline
32768 & 1.1741, 1.406e-4 & 1.1506, 1.421e-4 & 1.1158, 1.545e-4 & 1.0941, 2.133e-4 & 1.0835, 6.561e-4 \\
\hline
65536 & 1.1752, 1.497e-4 & 1.1508, 1.516e-4 & 1.1158, 1.722e-4 & 1.0949, 2.957e-4 & 1.0845, 0.0012 \\
\hline
131072 & 1.1749, 1.739e-4 & 1.1507, 1.817e-4 & 1.1161, 2.147e-4 & 1.0953, 4.776e-4 & 1.0837, 0.0023 \\
\hline
262144 & 1.1753, 2.244e-4 & 1.1509, 2.491e-4 & 1.1161, 3.449e-4 & 1.0952, 8.763e-4 & 1.0847, 0.0045 \\
\hline
524288 & 1.1752, 3.346e-4 & 1.1510, 4.006e-4 & 1.1161, 6.782e-4 & 1.0955, 0.0017 & 1.0848, 0.0090 \\
\hline
1048576 & 1.1753, 5.733e-4 & 1.1512, 7.311e-4 & 1.1161, 0.0014 & 1.0956, 0.0040 & 1.0846, 0.0205 \\
\hline
2097152 & 1.1753, 0.0011 & 1.1512, 0.0014 & 1.1162, 0.0030 & 1.0955, 0.0092 & 1.0848, 0.0494 \\
\hline
4194304 & 1.1754, 0.0020 & 1.1511, 0.0028 & 1.1162, 0.0061 & 1.0954, 0.0196 & 1.0849, 0.1077 \\
\hline
\multicolumn{6}{|c|}{Partition Size = 1024} \\
\hline
32768 & 1.0868, 1.426e-4 & 1.0749, 1.460e-4 & 1.0577, 1.657e-4 & 1.0474, 2.517e-4 & 1.0417, 8.630e-4 \\
\hline
65536 & 1.0872, 1.529e-4 & 1.0750, 1.585e-4 & 1.0580, 1.955e-4 & 1.0474, 3.695e-4 & 1.0426, 0.0017 \\
\hline
131072 & 1.0874, 1.872e-4 & 1.0754, 1.986e-4 & 1.0578, 2.617e-4 & 1.0476, 6.290e-4 & 1.0419, 0.0031 \\
\hline
262144 & 1.0874, 2.620e-4 & 1.0755, 3.021e-4 & 1.0581, 4.392e-4 & 1.0477, 0.0012 & 1.0423, 0.0061 \\
\hline
524288 & 1.0875, 4.250e-4 & 1.0754, 5.440e-4 & 1.0580, 9.832e-4 & 1.0476, 0.0025 & 1.0423, 0.0128 \\
\hline
1048576 & 1.0876, 7.775e-4 & 1.0755, 0.0011 & 1.0580, 0.0022 & 1.0476, 0.0068 & 1.0423, 0.0367 \\
\hline
2097152 & 1.0876, 0.0015 & 1.0755, 0.0021 & 1.0580, 0.0048 & 1.0477, 0.0156 & 1.0424, 0.0866 \\
\hline
4194304 & 1.0876, 0.0029 & 1.0755, 0.0042 & 1.0580, 0.0098 & 1.0477, 0.0330 & 1.0425, 0.1851 \\
\hline
\multicolumn{6}{|c|}{Partition Size = 2048} \\
\hline
32768 & 1.0433, 1.467e-4 & 1.0369, 1.525e-4 & 1.0287, 1.849e-4 & 1.0232, 3.346e-4 & 1.0196, 0.0014 \\
\hline
65536 & 1.0433, 1.604e-4 & 1.0376, 1.713e-4 & 1.0289, 2.383e-4 & 1.0235, 5.376e-4 & 1.0206, 0.0028 \\
\hline
131072 & 1.0434, 2.053e-4 & 1.0374, 2.206e-4 & 1.0290, 3.436e-4 & 1.0236, 9.811e-4 & 1.0210, 0.0052 \\
\hline
262144 & 1.0436, 3.185e-4 & 1.0376, 3.816e-4 & 1.0290, 6.051e-4 & 1.0238, 0.0019 & 1.0214, 0.0102 \\
\hline
524288 & 1.0437, 5.844e-4 & 1.0377, 7.747e-4 & 1.0290, 0.0015 & 1.0238, 0.0042 & 1.0211, 0.0221 \\
\hline
1048576 & 1.0437, 0.0012 & 1.0377, 0.0016 & 1.0290, 0.0037 & 1.0238, 0.0121 & 1.0211, 0.0677 \\
\hline
2097152 & 1.0437, 0.0023 & 1.0377, 0.0033 & 1.0290, 0.0079 & 1.0238, 0.0275 & 1.0213, 0.1570 \\
\hline
4194304 & 1.0437, 0.0046 & 1.0377, 0.0067 & 1.0291, 0.0164 & 1.0238, 0.0579 & 1.0213, 0.3331 \\
\hline
\end{tabular}
    \label{tab:metropolisC2rv}
\end{table*}

\end{document}